\documentclass[11pt]{article}

\usepackage[english]{babel}

\usepackage[letterpaper,top=2cm,bottom=2cm,left=3cm,right=3cm,marginparwidth=1.75cm]{geometry}

\usepackage{amsmath}
\usepackage{graphicx}
\usepackage[colorlinks=true, allcolors=blue]{hyperref}
\usepackage[title]{appendix}

\usepackage{amsthm,amssymb}

\usepackage{algorithm,algcompatible,amsmath}
\algnewcommand\INPUT{\item[\textbf{Input:}]}%
\algnewcommand\OUTPUT{\item[\textbf{Output:}]}%
\usepackage[utf8]{inputenc}

\usepackage{enumitem}
\usepackage{pdfpages}
\usepackage{blkarray}
\usepackage{mathtools}
\usepackage{mathtools}
\usepackage{setspace}
\usepackage{url}
\usepackage{natbib}
\usepackage{float}
\usepackage[capitalize]{cleveref}
\usepackage{verbatim}
\usepackage{rotating}

\def\bY{\mathbf Y}
\def\bF{\mathbf F}
\def\bX{\mathbf X}

\def\bZ{{\bf Z}}

\def\bC{{\bf C}}

\def\bA{\mathbf A}
\def\bB{\mathbf B}
\def\bU{\mathbf U}

\def\bW{\mathbf W}

\def\bX{\mathbf X}

\def\bI{\mathbf I}
\def\bJ{\mathbf J}

\def\bR{\mathbf R}
\def\bS{\mathbf S}

\def\bV{{\bf V}}
\def\bQ{\mathbf Q}
\def\bD{\mathbf D}
\def\bE{\mathbf E}
\def\bL{\mathbf L}
\def\bT{\mbox{\boldmath{$\Sigma_{n}$}}}

\usepackage{color}
\usepackage{comment}
\usepackage{ulem}
\usepackage{soul}

\newtheorem{thm}{Theorem}
\newtheorem{lem}{Lemma}
\newtheorem{prop}{Proposition}

\providecommand{\keywords}[1]
{
  \textbf{\textit{Keywords:}} #1
}

\title{Multiple Augmented Reduced Rank Regression for Pan-Cancer Analysis}

\author{Jiuzhou Wang, Eric F. Lock\\
Division of Biostatistics, University of Minnesota, Minneapolis, MN, USA \\ }

\begin{document}

\maketitle

\begin{abstract}
Statistical approaches that successfully combine multiple datasets are more powerful, efficient, and scientifically informative than separate analyses. To address variation architectures correctly and comprehensively for high-dimensional data across multiple sample sets (i.e., cohorts), we propose multiple augmented reduced rank regression (maRRR), a flexible matrix regression and factorization method to concurrently learn both covariate-driven and auxiliary structured variation.  We consider a structured nuclear norm objective that is motivated by random matrix theory, in which the regression or factorization terms may be shared or specific to any number of cohorts.  Our framework subsumes several existing methods, such as reduced rank regression and unsupervised multi-matrix factorization approaches, and includes a promising novel approach to regression and factorization of a single dataset (aRRR) as a special case.  Simulations demonstrate substantial gains in power from combining multiple datasets, and from parsimoniously accounting for all structured variation.    We apply maRRR to gene expression data from multiple cancer types (i.e., pan-cancer) from TCGA, with somatic mutations as covariates. The method performs well with respect to prediction and imputation of held-out data, and provides new insights into mutation-driven and auxiliary variation that is shared or specific to certain cancer types.\\ %An R implementation is available at \url{https://github.com/JiuzhouW/maRRR}. \\
\end{abstract}


\clearpage
\section{Introduction}

The proliferation of omics data in biomedicine and genomics has allowed for increasingly comprehensive investigations that span multiple sample sets and multiple molecular facets. Statistical approaches that successfully combine multiple datasets within a single analytical framework are more powerful, efficient, and scientifically informative than separate analyses. This has spurred advances in methodology for high-dimensional data integration, however, there remain unmet needs especially for multi-cohort data in which the same features are measured for different sample groups. \textcolor{black}{Our motivating example is gene expression and somatic mutation data from the Cancer Genome Atlas (TCGA) Pan-Cancer Project \citep{hoadley2018cell, Hunter2021TCGA}, for 6581 tumor samples from 30 cohorts corresponding to different cancer types. Given the importance of gene expression in the behavior of cancer, and the related etiology of distinct cancer types through somatic mutations,  we are interested in distinguishing variation due to somatic mutations from auxiliary structured variation in cancer gene expression and whether these effects are shared across cancer types.} 

Several unsupervised multi-matrix factorization methods provide low-rank representations of underlying structure. The singular value decomposition (SVD), principle component analysis (PCA) and other well-known %other  
approaches allow a %low-rank representation 
rank $r$ approximation of a single matrix %$\bX_{m \times n} \approx \bU_{m \times r} \bV_{n \times r}^T, r < \min{(m,n)}$. 
$\bX_{p \times n} \approx \bU_{p \times r} \bV_{n \times r}^T, r < \min{(p,n)}$. 
Loadings $\bU$ %(i.e. rows) 
and scores $\bV$ %(i.e. columns) 
explain variation in the rows or columns, respectively. The joint and individual variation explained (JIVE) method extends PCA to %multi-source data via $\bX_i = \bU_i \bV^T + \bW_i \bV_i^T + \bE_i$ \citep{Lock2013JIVE}. 
multiple datasets with shared columns $\{\bX_1,\hdots,\bX_J\}$ via $\bX_i \approx \bU_i \bV^T + \bW_i \bV_i^T$.
%Besides individual scores $\bV_i$, the joint scores $\bV$  allows an integration analysis of shared structure among the datasets. 
Here the joint scores $\bV$ capture shared structure among the datasets, and the individual scores $\bV_i$ capture structure specific to dataset $i$. Numerous related approaches, such as AJIVE \citep{Feng2018AJIVE} and SLIDE \citep{Irina2019SLIDE} %and \textcolor{black}{BIDIFAC \citep{park2020integrative}}
 have been proposed to factorize multiple data from other perspectives. Moreover, BIDIFAC+ \citep{Lock2022BADIFACplus} enables a more flexible way to identify multiple shared and specific modules of variation%. Constructing different combinations of data gives different partially-joint scores.
, which may be partially shared over row subset or column subsets.  However, these unsupervised methods suffer from neglecting covariate information. Other supervised techniques \citep{ Wang2021DeepIDA, zhang2022joint} %Safo2021SIDA,
identify structures across multiple datasets relevant to predicting an outcome, but they do not capture both covariate-driven and auxiliary structures. %. However, they do not detect low-rank decomposition of original data, let alone reveal variation signals. They are designed for discrete group outcomes as well. Thus, there is a must to establish a flexible supervised low-rank factorization for continuous outcomes. 

To impose low-rank covariate effects, different types of penalties have been introduced in the the multivariate least square regression framework. Reduced rank regression (RRR) \citep{Izenman1975RRR} is a popular %dimension reduction approach for one single matrix. 
approach to predict $\bX: p \times n$ from $\bY: q \times n$ via least squares in which the coefficients have low-rank, $\bX \approx \bB \bY$ with rank$(\bB) < \min{(p,q)}$.
%It will construct latent factors in the predictor space by adding a low-rank restriction to the multivariate least-squares problem. %approximation Frobenius norm.
Rank penalized (RSC) \citep{Bunea2011opti} and nuclear-norm penalized (NNP) least square criteria \citep{Yuan2007NN} are widely used %versions with similar penalties enforcing low ran same time. 
alternatives with penalties that enforce low-rank coefficients. Combining RRR with adaptive NNP \citep{Chen2013ANN} shows a better performance than RSC. Integrative RRR \citep{Li2019iRRR} extends the estimation to multiple covariate sets all at once. Nonetheless, those regression methods have two limitations: (1) they do not allow for potentially unique covariate-driven signals across multiple sample cohorts and (2) they do not account for additional low-rank structure unrelated to the covariates.

Missing values %, especially in genomics data, %such as TCGA, %exist a lot for different causes: inappropriate lab operations, misconduct in the data collection process, systematic errors and etc. 
occur in genomics and other fields due to cost limitations or other technical issues.  The data may have three types of missingness: entry-wise, column-wise or row-wise. To impute missing values %with global structures, 
matrix factorization based approaches, such as SVDImpute \citep{troyanskaya2001missing} and SoftImpute \citep{mazumder2010spectral} are popular since they are effective and straightforward, and many of the the aforementioned methods can be modified for imputation. However, %they do not show satisfactory performances in column and row missing as a result of partially-shared and individual structures. 
they will suffer from the same limitations described above. 

Unifying reduced rank regression and unsupervised low-rank factorization using the nuclear norm penalty, we develop the multiple augmented reduced rank regression (maRRR) method for multi-cohort data that enables a very flexible approach for the simultaneous identification of covariate-driven effects and auxiliary structured variation. These covariate effects and augmented structures may be shared across any cohorts via a general objective function. %Compared with the previous unsupervised methods, this novel supervised low-rank factorization accomplish the ensuing tasks: a) more accurate predictions for future samples; b) more precise and robust missing data imputation under various circumstances of dearth of data; c) dominant variation structure uncovering across cohorts and signal diagnosis in both covariate-related and covariate-unrelated sections. 
This novel low-rank regression and factorization method can be used to impute various types of missing data, accurately capture the relationship between covariates and high-dimensional outcomes, and explore covariate-related and covariate-unrelated patterns of variation that are shared across or specific to different cohorts.  %Due to space limitations for this submission, proofs are omitted from the main article and are available in Web Appendix B.
%All proofs are omitted from the main article and are available as supplemental materials.

%Besides aforementioned multi-cohort data, this method is designed to analyze multi-view data (data on the same subjects from different sources) as well. This is achieved by simply switching the way we integrate matrices: from horizontally to vertically. Without loss of generality, the notations in the following are all in the context of multi-cohort data.

%%\vspace{-1.5cm}

\section{Proposed Model}

%\subsection{General formulation}
Let $\bX_j: p \times n_j$ denote data matrices with accompanying covariates $\bY_j: q \times n_j$ for $j$ sample cohorts   $j=1,...,J$. Concatenations across all cohorts are denoted by $\cdot$, e.g., $\bX_{\cdot} = [\bX_1,\bX_2,...,\bX_{J}]$ and $\bY_{\cdot} = [\bY_{1},\bY_{2},...,\bY_{J}]$. \textcolor{black}{Both $\bX_{\cdot}$ and $\bY_{\cdot}$ are the only observed data in the model.}  For our application, we consider gene expression data $\bX_{\cdot}$ and somatic mutations $\bY_{\cdot}$ for several patients across $J=30$ cancer types. %Based on these two types of observed data, 
\textcolor{black}{We are interested in decomposing $\bX_{\cdot}$ into `modules' of low-rank covariate-driven or auxiliary structures, where each module is shared on a different subset of the cohorts.} We estimate low-rank coefficient matrices $\bB_k: p \times q$ for \textcolor{black}{$k = 1,...,K$ modules of covariate-driven variation} and we concurrently estimate \textcolor{black}{low-rank} auxiliary variation structures $\bS_{\cdot}^{(l)}: p \times n$ \textcolor{black}{for  $l = 1,...,L$ modules.} Acknowledging the errors $\bE_j: p \times n_j , j=1,...,J$ for each cohort, the full model is% can be written as:
\begin{align}\label{model}
    \bX_{\cdot} &= \sum^{K}_{k=1} \bB_k \bY_{\cdot}^{(k)} + \sum^{L}_{l=1}  \bS_{\cdot}^{(l)} + \bE_{\cdot} 
    \end{align}
    where $\bY_{\cdot}^{(k)} = [\bY_{1}^{(k)},\bY_{2}^{(k)},...,\bY_{J}^{(k)}]$, $\bS_{\cdot}^{(l)} = [\bS_{1}^{(l)},\bS_{2}^{(l)},...,\bS_{J}^{(l)}]$, $\bE_{\cdot} = [\bE_{1},\bE_{2},...,\bE_{J}]$.

The presence of each $\bY_{j}^{(k)}$ or $\bS_{j}^{(l)}$ across the cohorts are determined by binary indicator matrices $\bC_Y: J\times K$ and $\bC_S: J\times L$ respectively: 
\begin{equation*}
      \bY_{j}^{(k)}= \begin{dcases} 
 \boldsymbol{0}_{q\times n_j} & \text{if }  \bC_Y[j,k] = 0 \\ 
 \bY_j & \text{if }  \bC_Y[j,k] = 1 ,
  \end{dcases} 
    \bS_{j}^{(l)} = 
  \begin{dcases} 
\boldsymbol{0}_{p\times n_j} & \text{if }  \bC_S[j,l] = 0 \\ 
 \bU_S^{(l)}\bV_{Sj}^{(l)T} & \text{if }  \bC_S[j,l] = 1.
  \end{dcases}
\end{equation*}
$\bU_S^{(l)}$ represents shared loadings and $\bV_{Sj}^{(l)}$ sample scores for cohort $j$ in module $l$. Both indicator matrices may be determined either by pre-existing knowledge or via a data-driven algorithm, which we will detail in Section~\ref{sec:decomposition} and Appendix~\ref{app:indicator}. They are fixed in the model estimation process. There should be no identical columns within $\bC_Y$, so that each $\bB_k\bY_{\cdot}^{(k)}$ is present on a distinct subset of the cohorts. Similarly, no duplicate columns within $\bC_S$.
We refer to each $\bB_k\bY_{\cdot}^{(k)}$ and $\bS_{\cdot}^{(l)}$ as a module. 
Each $\bS_{\cdot}^{(l)}$ gives a low-rank module that explains covariate-unrelated structured variability within the cohorts (e.g., cancer types) identified by $\bC_S[:,l]$. Each $\bB_k \bY_{\cdot}^{(k)}$ gives another low-rank module for covariate-driven structure for the cancer type identified by $\bC_Y[:,l]$. Each module is assumed to be low-rank, meaning it can be factorized as the product of a small number of row and column vectors, $\bB_k = \bU_B^{(k)} \bV_{B}^{(k) T}$ and  $\bS_{\cdot}^{(l)} = \bU_S^{(l)} \bV_{S}^{(l) T}$. \textcolor{black}{We provide a schematic of our model in \cref{fig:model} and a table of notation details in Appendix~\ref{app:notation}.}

\begin{figure}[]
        \centering
        \includegraphics[height=19cm,width=15cm]{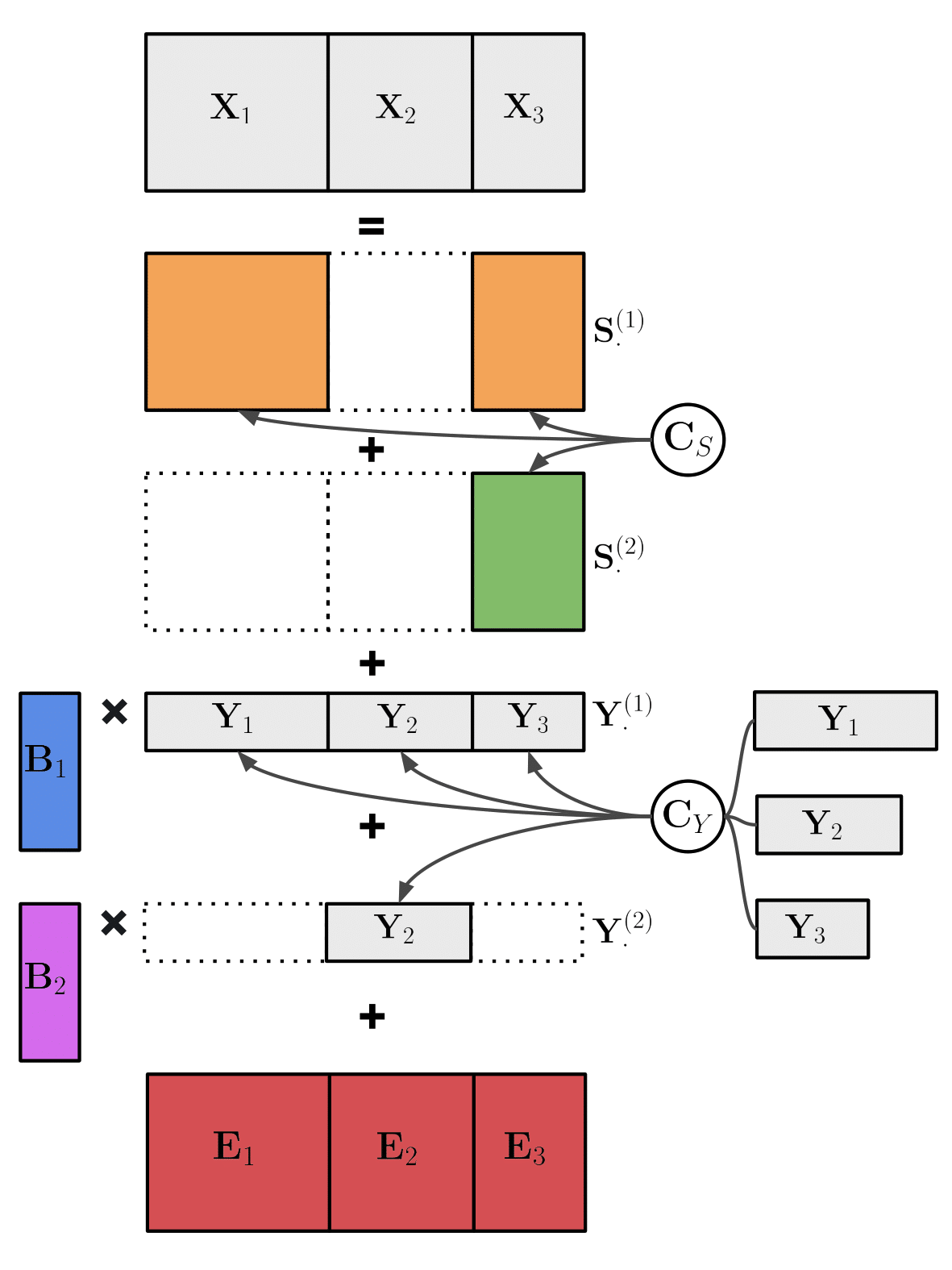}
        \caption{A schematic of our proposed model maRRR with 3 cohorts as an example. All matrices in grey are observed, i.e. outcomes $\bX_{\cdot} = [\bX_1,\bX_2,\bX_3]$ and covariates $\bY_{\cdot} = [\bY_1,\bY_2,\bY_3]$. Two binary indicator matrices for auxiliary structures $\bC_S = [[1,0,1]^T,[0,0,1]^T]$ and for covariate effects $\bC_Y = [[1,1,1]^T,[0,1,0]^T]$ are pre-specified. Then, the structures of $\bS_{\cdot}^{(1)} = [\bU_S^{(1)}\bV_{S1}^{(1)T},\mathbf{0},\bU_S^{(1)}\bV_{S3}^{(1)T}]$, $\bS_{\cdot}^{(2)} = [\mathbf{0},\mathbf{0},\bU_S^{(2)}\bV_{S3}^{(2)T}]$ and $\bY_{\cdot}^{(1)} = [\bY_1,\bY_2,\bY_3], \bY_{\cdot}^{(2)} = [\mathbf{0},\bY_2,\mathbf{0}]$ are determined. All matrices in color are to estimate, i.e., auxiliary structures $\bS_{\cdot}^{(1)}, \bS_{\cdot}^{(2)}$, covariate effect coefficients $\bB_1,\bB_2$ and random errors  $\bE_{\cdot} = [\bE_1,\bE_2,\bE_3]$.} %The full model is $\bX_{\cdot} = \bB_1 \bY_{\cdot}^{(1)} + \bB_2 \bY_{\cdot}^{(2)} +  \bS_{\cdot}^{(1)} + \bS_{\cdot}^{(2)} + \bE_{\cdot}$}
        \label{fig:model}
\end{figure}

%Either $\bY^{(k)}$ or $\bS^{(l)}$ will be non-zero when at least one $\bC_Y[j,k] = 1$ or $\bC_S[j,l] = 1, \forall j = 1,...,J$. There are $2^K-1$ such matrices for $\bY$ and $2^L -1$ such matrices for $\bS$ so in totally $2^L(2^K-1)$ enumerations for all possible modules, including cases that only having $\bY$. 

%\subsection{Special cases}

%There are three special cases of the general formulation worth noticing:
%\begin{enumerate}
%    \item (novel) Augmented reduced rank regression (aRRR), our proposed approach minimizing $ K= L = J = 1$. The reduced rank regression model is “augmented" to account for auxiliary structured variation $\bS$. The model is specialized for one single cohort and can be written as $\bX = \bB \bY + \bS + \bE$.
%    \item (novel) Multi-cohort reduced rank regression (mRRR), our proposed approach minimizing $L = 0$, i.e. no auxiliary terms $\bS$. The formulation $\bX_{\cdot} &= \sum^{K}_{k=1} \bB_k \bY_{\cdot}^{(k)} + \bE_{\cdot} $ can effectively recover shared and cohort specific covariate effects.
%    \item (existing) BIDIFAC+ \citep{Lock2022BADIFACplus} with horizontal structures only, our proposed approach minimizing $K = 0$, i.e. no covariate effects $\bB \bY$. The formulation $\bX_{\cdot} &=  \sum^{L}_{l=1}  \bS_{\cdot}^{(l)} + \bE_{\cdot} $ counts all shared and cohort specific variation together unsupervised.
%\end{enumerate}

%%\vspace{-1cm}

\section{Objective Function} \label{sec:obejective}

To estimate model \eqref{model} and impose low-rank structure, we minimize the following least squares criterion with a structured nuclear norm penalty: %denoted as $f_1(\{\bB_k\}^K_{k=1},\{\bS_{\cdot}^{(l)}\}^L_{l=1})$:
\begin{align}\label{ObjZ}
    \min_{\{\bB_k\}^K_{k=1},\{\bS_{\cdot}^{(l)}\}^L_{l=1}} \{ \frac{1}{2} ||\bX_{\cdot} - \sum^{K}_{k=1} \bB_k \bY_{\cdot}^{(k)} - \sum^{L}_{l=1}  \bS_{\cdot}^{(l)}||_F^2 + \sum^{K}_{k=1} \lambda_B^{(k)} ||\bB_k||_{*} +\sum^{L}_{l=1} \lambda_S^{(l)} ||\bS_{\cdot}^{(l)}||_{*}  \}
\end{align}
Here $||\cdot||_*$ denotes the nuclear norm, i.e., the sum of the singular values of the matrix, a convex penalty which encourages a low-rank solution. 
There are three special cases of the general objective functions worth noting. The first two are novel and the third is previously described, listed as follows:
\begin{enumerate}
    \item Augmented reduced rank regression (aRRR), our proposed approach minimizing \eqref{ObjZ} for $ K= L = J = 1$. The nuclear-norm penalized reduced rank regression model for a single matrix is “augmented" to account for auxiliary structured variation $\bS$ simultaneously. 
    \item Multi-cohort reduced rank regression (mRRR), our proposed approach minimizing \eqref{ObjZ} for $L = 0$, i.e. no auxiliary terms $\bS$. The reduced rank regression is extended to recover multiple (shared or individual) covariate effects at once.
    \item Optimizing this objective with no covariate-driven structure ($K=0$) corresponds to the BIDIFAC+ method \citep{Lock2022BADIFACplus} with horizontal structures only.
\end{enumerate}

%\begin{lem}\label{ZtoUV}
%(Mazumder et al., 2010)
%For any matrix $\bZ: m\times n$ with   $rank(\bZ) = k $, $\forall  r\geq k$, the following holds:
%\begin{align*}
%    || \bZ ||_{*} & = \min_{\substack{\bU,\bV:\\\bZ = \bU_{m\times r}\bV_{n\times r}^T}} \frac{1}{2} (|| \bU ||_{F}^2 + || \bV ||_{F}^2)
%\end{align*}
%\end{lem}

%The proof is given in the original paper \cite{Mazumder2010spectral}. The lemma not only leverages nuclear norm and matrix decomposition but also points out that, there is no further requirement on the number of columns $r$ of the factor matrices $\bU$ and $\bV$ as long as it is greater than rank of $\bZ$. Therefore, if replacing nuclear norm penalty with this Frobenius norm in optimization, we do not need to specify the true rank of estimated matrix $\bZ$ but instead give an upper bound $r\geq rank(\bZ)$. This motives us to use nuclear norm penalty with matrix decomposition to impose low-rank estimations.
%By \cref{ZtoUV}, 
\cite{mazumder2010spectral} and others have noted the equivalence of the nuclear norm penalty and an additive $L_2$ penalty on the terms in the low-rank factorization, and this leads to an alternative form of our objective (\cref{ObjZ}),% denoted as $f_2(\{\bU_B^{(k)},\bV_B^{(k)}\}^K_{k=1}, \{\bU_S^{(l)},\bV_S^{(l)}\}^L_{l=1})$:
\begin{align}\label{ObjUV}
    \min_{\substack{\{\bU_B^{(k)},\bV_B^{(k)}\}^K_{k=1}, \{\bU_S^{(l)},\bV_S^{(l)}\}^L_{l=1}}}  &\frac{1}{2}\{ ||\bX_{\cdot} - \sum^{K}_{k=1} \bU_B^{(k)}\bV_B^{(k)T} \bY_{\cdot}^{(k)} - \sum^{L}_{l=1}  \bU_S^{(l)}\bV_S^{(l)T}||_F^2 + \nonumber\\
    &\sum^{K}_{k=1} \lambda_B^{(k)} (||\bU_B^{(k)}||_{F}^2 +||\bV_B^{(k)}||_{F}^2) +\sum^{L}_{l=1} \lambda_S^{(l)}( ||\bU_S^{(l)}||_{F}^2 + ||\bV_S^{(l)}||_{F}^2)  \}
\end{align}
where we only need to set a general upper bound for the estimated rank of each $\bB_k$ and $\bS_{\cdot}^{(l)}$, i.e. $r_{B,upper}$ and $r_{S,upper}$. \textcolor{black}{These upper bounds serve as the number of columns for each $\bU_B^{(k)}$, $\bV_B^{(k)}$, $\bU_S^{(l)}$, and $\bV_S^{(l)}$. The actual ranks of the solution may be smaller due to the rank sparsity encouraged by the nuclear norm penalty, and if the upper bounds are large enough the solution will correspond to that in equation~\eqref{ObjZ}.  We state this formally in Theorem~\ref{equ}}; the proof of this result, and all other novel results in this manuscript, are given in Appendix~\ref{app:proof}.

%\vspace{-0.3cm}
\begin{thm}\label{equ}
 %For all penalty terms $\lambda>0$, the solutions to (\ref{ObjZ}) and (\ref{ObjUV}) coincide.
 If both (\ref{ObjZ}) and (\ref{ObjUV}) have the same penalty terms $\lambda_B^{(k)} >0, k= 1,...,K$ and $\lambda_S^{(l)} >0, l= 1,...,L$, the solutions to the objective functions coincide.
\end{thm}
%%\vspace{-0.15cm}

In what follows in Section~\ref{sec:theory} we describe a random matrix theory approach to automatically select the nuclear norm penalty weights $\lambda$ for the different modules.

%%\vspace{-1.2cm}

\section{Theoretical results}
\label{sec:theory}

We describe conditions on the penalty to avoid degenerate cases in which certain modules are guaranteed to be zero in the solution (regardless of the data $\bX_{\cdot}$ and $\bY_{\cdot}$) in Proposition~\ref{lambda_selection}. 
\begin{prop}\label{lambda_selection}
The following conditions on penalty parameters  are needed to allow for non-zero estimation of each $\{\bB_k\}_{k=1}^K,\{\bS^{(l)}\}_{l=1}^L$:
\begin{enumerate}
  %  \item For $ k\neq k'$, if there exists a module $\bY_{\cdot}^{(k')}$ is contained in another module $\bY_{\cdot}^{(k)}$, i.e. $\bC_Y[j,k]\geq\bC_Y[j,k'],\forall j$, then $\lambda_B^{(k')}< \gamma_Y\lambda_B^{(k)}$, where $\gamma_Y$ is a constant dependent on $\bY$ and specified in the proof for brevity.
    \item Let $\mathcal{I}_k \subset \{1,...,k-1,k+1,...,K \}$ be any subset of $\bY$ modules for which the non-zero blocks of $\{\bY_{\cdot}^{(i)}\}_{i\in \mathcal{I}_k} $ cover exactly those of $\bY^{(k)}$, i.e. $\sum_{i\in \mathcal{I}_k} \bC_Y[\cdot,i] = c_y \cdot \bC_y[\cdot,k]$ for some positive integer $c_y$.  Then, $\lambda_B^{(k)} < \frac{1}{c_y} \sum_{i\in \mathcal{I}_k} \lambda_B^{(i)}$. %Note there may be multiple $\mathcal{I}_k$'s in one objective function.
    \item Let $\mathcal{I}_k \subset \{1,2,...,L \}$ be any subset of $\bS$ modules for which the non-zero blocks of $\{\bS_{\cdot}^{(i)}\}_{i\in \mathcal{I}_k} $ cover exactly those of $\bY_k$, i.e. $\sum_{i\in \mathcal{I}_k} \bC_S[\cdot,i] = c_{sy} \cdot \bC_Y[\cdot,k]$ for some positive integer $c_{sy}$. Then, $\lambda_B^{(k)} < \frac{1}{c_{sy}} \sum_{i\in \mathcal{I}_k} \lambda_S^{(i)} ||\bY_{\cdot}^{(k)}||_{*}$. %Note there may be multiple $\mathcal{I}_k$'s in one objective function.
   % \item Let $\mathcal{I}_l \subset \{1,2,...,K \}$ be any subset of $\bY$ modules that the non-zero blocks of $\{\bY_{\cdot}^{(i)}\}_{i\in \mathcal{I}_l} $ cover exactly those of $\bS_{\cdot}^{(l)}$, i.e. $\sum_{i\in \mathcal{I}_l} \bC_Y[\cdot,i] = c_{ys} \cdot \bC_S[\cdot,l]$ for some positive integer $c_{ys}$, then\\ $\lambda_S^{(l)} < \frac{1}{c_{ys}} \sum_{i\in \mathcalx{I}_l} \lambda_B^{(i)}||\bY_{\cdot}^{(i)T}(\bY_{\cdot}^{(i)}\bY_{\cdot}^{(i)T})^{-1}||_{*}$. Note there may be multiple $\mathcal{I}_l$'s in one objective function.
    \item For $ l\neq l'$, if a module $\bS_{\cdot}^{(l')}$ is contained in another module $\bS_{\cdot}^{(l)}$, i.e. $\bC_S[j,l]\geq\bC_S[j,l'],\forall j$, then $\lambda_S^{(l')}< \lambda_S^{(l)}$.
    \item Let $\mathcal{I}_l \subset \{1,...,l-1,l+1,...,L \}$ be any subset of $\bS$ modules that the non-zero blocks of $\{\bS_{\cdot}^{(i)}\}_{i\in \mathcal{I}_l} $ cover exactly those of $\bS_{\cdot}^{(l)}$, i.e. $\sum_{i\in \mathcal{I}_l} \bC_S[\cdot,i] = c_s \cdot \bC_S[\cdot,l]$ for some positive integer $c_s$. Then, $\lambda_S^{(l)} < \frac{1}{c_s} \sum_{i\in \mathcal{I}_l} \lambda_S^{(i)}$. %Note there may be multiple $\mathcal{I}_l$'s in one objective function.
\end{enumerate}
\end{prop}
%\vspace{-0.5cm}
To motivate a random matrix theory approach to select the tuning parameters, we present two results establishing the connection between the nuclear norm penalty and singular value thresholding.  Lemma~\ref{nuclearA} is a well-known result for the unsupervised case \citep{Cai2010ASV}, and in Proposition~\ref{nuclearAY} we extend it to the regression context.    
\begin{lem}\label{nuclearA}
%(Cai et al., 2010)
Let $\bU\bD\bV^T$ be the SVD of a matrix $\bX$. The solution to $\min_{\bS} \{\frac{1}{2} ||\bX-\bS||_F^2 + \lambda||\bS||_{*}\}$ is $\bS = \bU\widetilde{\bD}\bV^T$, where $\widetilde{\bD}$ is diagonal with entries $\widetilde{\bD}[i,i] = \max(\bD[i,i]-\lambda,0)$. 
\end{lem}

%\vspace{-0.5cm}

\begin{prop}\label{nuclearAY}
Let $\bY$ be a semi-orthogonal matrix such that $\bY\bY^T = \bI$ and $\bU\bD\bV^T$ be the SVD of a matrix $\bX\bY^T$. The solution to both of the following objectives:
\begin{align*}
   \min_{\bB} \{\frac{1}{2} ||\bX-\bB\bY||_F^2 + \lambda||\bB||_{*}\} \; \; \text{and} \; \; \min_{\bB} \{\frac{1}{2} ||\bX-\bB\bY||_F^2 + \lambda||\bB\bY||_{*}\},
\end{align*}
is $\bB = \bU\widetilde{\bD}\bV^T$, where $\widetilde{\bD}$ is diagonal with entries $\widetilde{\bD}[i,i] = \max(\bD[i,i]-\lambda,0)$. %This is also the solution to
%\begin{align*}
  % \min_{\bA} \{\frac{1}{2} ||\bX-\bA\bY||_F^2 + \lambda||\bA\bY||_{*}\}.
%\end{align*}

\end{prop}

%\vspace{-0.3cm}
%When $\bY$ is the identity matrix, \cref{nuclearAY} degenerates to \cref{nuclearA}. 

%\Cref{nuclearAY} also ends the debate of choosing penalty terms \citep{Yuan2007NN}\citep{Chen2013ANN} by showing the equivalence of $||\bA||_*$ and $||\bA\bY||_*$ with some semi-orthogonal $\bY$ that $\bY\bY^T=\bI$. 
While the relative merits of penalizing  $||\bB||_*$ or $||\bB\bY||_*$ has been debated \citep{Yuan2007NN,Chen2013ANN}, \cref{nuclearAY} shows they are identical if $\bY$ is semi-orthogonal.
In practice, we orthogonalize the columns of $\bY$ prior to estimation.  However, this requires that the number of features in $\bY$ is less than the sample size (e.g., $q<n$); if $q \geq n$ then $\bY$ will be \textcolor{black}{semi-orthogonal in the opposite direction $\bY^T\bY = \bI$,} causing the solution to degenerate to the unsupervised case, which we establish in \textcolor{black}{Proposition 5 in Appendix~\ref{app:proof}}. 

The following propositions describe the distribution of the singular values of a random matrix under general assumptions, which can then be used to motivate tuning parameters.% widely-used and time-costing cross-validation grid searching method.

%\vspace{-0.3cm}
\begin{prop} \label{Rudelson}
Let $\lambda_{max}$ be the largest singular value of a matrix $\bE: m\times n$ of independent Guassian entries with mean 0 and variance $\sigma^2$.  \textcolor{black}{We have $E(\lambda_{max}) \leq \sigma(\sqrt{m}+\sqrt{n})$.}
%As $m,n \rightarrow \infty$,
\end{prop}

%\begin{prop} \label{Shabalin}
%Let $\bX,\bA,\bE$ be three $m\times n$ matrices such that $\bX=\bA+\bE$, where entries of $\bE$ are independent Guassian with mean 0 and variance $\sigma^2$. Assume $rank(\bA) = r$.  As $ m,n \rightarrow \infty$, sorted singular values $\lambda_1(\bX)\geq ... \geq \lambda_r (\bX) \geq \sigma(\sqrt{m}+\sqrt{n}) \geq \lambda_{r+1} (\bX) \geq... \geq\lambda_{min(m,n)}(\bX)>0$.
%\end{prop}

%\begin{prop} \label{shabalin_Y}
%Let $\bX,\bA,\bE$ be three matrices such that $\bX_{m\times n}=\bA_{m\times q}\bY_{q\times n}+\bE_{m\times n}$, where entries of $\bE$ are independent Guassian with mean 0 and variance $\sigma^2$. $\lambda = \sqrt{m} + \sqrt{q}$ is a reasonable choice for the objective function $\min_{\bA} \{\frac{1}{2} ||\bX-\bA\bY||_F^2 + \lambda||\bA||_{*}\}$. In particular, if $\bY$ is an identity matrix with $q=n$, $\lambda = \sqrt{m} + \sqrt{n}$ is a reasonable choice for the objective function $\min_{\bA} \{\frac{1}{2} ||\bX-\bA||_F^2 + \lambda||\bA||_{*}\}$.
%\end{prop}

%\vspace{-0.7cm}

\begin{prop} \label{Shabalin}
Let $\bY_{q\times n}$ be semi-orthogonal such that $\bY\bY^T = \bI$. For integers $m,q \geq1 $ defined in a way that $\frac{m}{q} \rightarrow c>0 $ as $q\rightarrow \infty$, Let $\bX_{m\times n},\bB_{m\times q},\bE_{m\times n}$ be three matrices such that $\bX=\bB\bY_{q\times n}+ \frac{1}{\sqrt{q}}\bE$, where entries of $\bE$ are independent Guassian with mean 0 and variance $\sigma^2$. Assume $rank(\bB) = r$. Denote the singular values of $\bB$ and $\bX\bY^T$ are $\sigma_1(\bB)\geq ... \geq \sigma_r (\bB)>0$ and $\sigma_1(\bX\bY^T)\geq ... \geq \sigma_r (\bX\bY^T) >0$ respectively. As $n \rightarrow \infty$, 
\begin{align*}
\sigma_j(\bX\bY^T)\xrightarrow{P}
    \begin{dcases} s(\sigma_j(\bB)) > 1 + \sqrt{c},
   & \text{if }  \sigma_j(\bB)>\sqrt[4]{c} \\ 
  1+\sqrt{c}, & \text{if } \sigma_j(\bB)\leq\sqrt[4]{c} 
  \end{dcases},\forall 1\leq j\leq r,
\end{align*}
where $s(\cdot)$ is a known function. In particular, when $\bY$ is an identity matrix ($q=n$) and $\bX=\bB+ \frac{1}{\sqrt{n}}\bE$, it follows that 
$\sigma_j(\bX)\xrightarrow{P}
    \begin{dcases} s(\sigma_j(\bB)) > 1 + \sqrt{c},
   & \text{if }  \sigma_j(\bB)>\sqrt[4]{c} \\ 
  1+\sqrt{c}, & \text{if } \sigma_j(\bB)\leq\sqrt[4]{c}. 
  \end{dcases}$
\end{prop}

%\vspace{-0.3cm}

\cref{Rudelson} comes directly from \citep{Rudelson2010NAT}, and \cref{Shabalin} is closely related to the result in \citep{Shabalin2013res}. Consider the reasonable penalty for $\bS$ in \cref{nuclearA}, i.e. $\bX_{m\times n}=\bS_{m\times n}+\bE_{m\times n}$. A set of reasonable tuning parameters will %1) detect the original signals other than noises. Since \cref{Rudelson} points out that the largest singular values of random Guassian noise matrix is $\leq \sigma(\sqrt{m}+\sqrt{n})$, this requires each signal's penalty to be $\geq \sigma(\sqrt{m}+\sqrt{n})$ in order to get rid of the largest noises. 2) not overlook any possible signals. 
(1) detect the low-rank signals and (2) not capture components that are solely due to noise.  Considering Propositions~\ref{Rudelson} and \ref{Shabalin}, setting $\lambda = \sigma(\sqrt{m}+\sqrt{n})$ is reasonable because it only keeps the signals (top $r$ components) whose singular values are expected to be greater than those of independent random noise.
%Since \cref{Shabalin} states that the smallest singular value of final $\bX$ is $\geq \sigma(\sqrt{m}+\sqrt{n})$, we should not give a penalty larger than this in order to detect the smallest signal. Luckily, only the choice of $\lambda = \sigma(\sqrt{m}+\sqrt{n})$ can satisfy both the requirements simultaneously.  Therefore, setting $\lambda = \sigma(\sqrt{m}+\sqrt{n})$ is reasonable because it only keeps the signals (top $r$ components) whose singular values are expected to be greater than those of independent random noises. 
Consider the reasonable penalty for $\bB$ in \cref{nuclearAY}, i.e. $\bX_{m\times n}=\bB_{m\times q}\bY_{q\times n}+\bE_{m\times n}$. Following a similar argument, we set $\lambda = \sigma(\sqrt{m}+\sqrt{q})$.

In practice, after normalizing raw data as described in Appendix~\ref{app:scaling}, the noise variance for $\bX_{\cdot}$ is $1$ ($\sigma = 1$) and each $\bY^{(k)}$ are semi-orthogonal.  Thus, in order to distinguish true signals $\{\bB_k\}^K_{k=1},\{\bS_{\cdot}^{(l)}\}^L_{l=1}$ from Gaussian noise in the objective (\ref{ObjZ}), we fix $\lambda_B^{(k)} = \sqrt{p}+\sqrt{q}$ for any module $\bB_k, k=1,...,K$ and $\lambda_S^{(l)} = \sqrt{p}+\sqrt{\sum_{j=1}^J n_j\bC_S[j,l]}$ for any module $\bS_{\cdot}^{(l)},l=1,...,L$. This directly extends our choices for a single matrix, as estimating any given module $\bB^{(k)}$ or $\bS_{\cdot}^{(l)}$ with the others fixed reduces to the setting of the previous propositions.  

%Admittedly, read-world data does not necessarily have orthogonality. However, we can orthogonalize the data first and then linear transform them back to their original scale after analysis.
%Note that propositions~\ref{sameBYS} and \ref{Shabalin} require orthogonality in $\bY$, which generally does not hold in practice.  However, we can still make use of these results by orthogonalizing the data and then transforming back to their original scale after analysis. 

%\vspace{-1.3cm}

\section{Estimation}

%%\vspace{-0.2cm}

In practice we scale $\bX$ \citep{gavish2017optimal} and orthogonalize $\bY$ prior to optimization. 
%We first center each row of $\bX_{\cdot}$ to have mean 0. In order to satisfy the standard normal noise requirement, we estimate the error variance for $\bX_{\cdot}$ by using the median absolute deviation estimator from \citep{gavish2017optimal}. The estimated variance of $\bX_{\cdot}$ is denoted as $\hat{\sigma}^2$. Then, we use $\bX_{\cdot}/\hat{\sigma}$ as the final data matrix for optimization, which has residual variance approximately 1. 
%We further orthogonalize the columns of $\bY$ via SVD prior to optimization, and transform the solution back to the original covariate space afterward;  
The details of this procedure are provided in Appendix~\ref{app:scaling}, and a simulation study illustrating its advantages is provided in Appendix~\ref{app:additional_sims}.
%Before applying optimization steps to the proposed loss function, we have two types of scaling for the observations, standardization and orthogonalization. Details can be found in Appendix.

%\vspace{-0.9cm}

\subsection{Optimization}
\label{sec:opt}
 We estimate all regression coefficients $\bB$ and auxiliary variation sources $\bS$ simultaneously, via alternating optimization approaches for either formulation \eqref{ObjZ} or \eqref{ObjUV} of our objective. %by iterativily solving objective function (\ref{ObjUV}). 
 For objective \eqref{ObjUV}, the introduction of $\bU$ and $\bV$ can make the optimization algorithm more efficient because the objective function has a closed-form gradient. Given all other estimates, we update every single $\bU_B^{(k)},\bV_B^{(k)},\bU_S^{(l)},\bV_S^{(l)}$ by setting its corresponding gradient to be zero. The details are provided in Algorithm 1.  
\begin{algorithm}[]
    \caption{Alternating Least Square with Matrix Decomposition} \label{ALSUV}
  \begin{algorithmic}[1]
    \INPUT Covariates $\bY$ and corresponding multivariate outcomes $\bX$; penalizing terms $\lambda_B, \lambda_S$; binary indicator matrices $\bC_Y, \bC_S$
    \OUTPUT $\bB,\bS$
    \STATE \textbf{Initialization} Construct $\{\bY_{\cdot}^{(k)}\}_{k=1}^K$ based on $\bC_Y$. Assign initialized numbers for each entry of $\{\bU_B^{(k)},\bV_B^{(k)}\}^K_{k=1}, \{\bU_S^{(l)},\bV_S^{(l)}\}^L_{l=1}$ 
    \WHILE{convergence criterion does not meet}
    \FOR{$k=1,...,K$}
      \STATE Compute the residual matrix $\bX_{\cdot}^{(k)} = \bX_{\cdot} - \sum^{K}_{k'=1,k'\neq k} \bU_B^{(k')}\bV_B^{(k')T} \bY_{\cdot}^{(k')} - \sum^{L}_{l=1}  \bU_S^{(l)}\bV_S^{(l)T}$
      \STATE Update $\bU_B^{(k)} = \bX_{\cdot}^{(k)}\bY_{\cdot}^{(k)T}\bV_B^{(k)}(\bV_B^{(k)T}\bY_{\cdot}^{(k)}\bY_{\cdot}^{(k)T}\bV_B^{(k)} + \lambda_B^{(k)}\bI_{r_B})^{-1}$
      \STATE Update $vec(\bV_B^{(k)}) = [(\bU_B^{(k)T}\bU_B^{(k)}) \bigotimes (\bY_{\cdot}^{(k)}\bY_{\cdot}^{(k)T}) + \lambda_B^{(k)}\bI_{q * r_B} ]^{-1} vec[\bY_{\cdot}^{(k)}(\bX_{\cdot}^{(k)T})\bU_B^{(k)}]$
      \STATE Transform $vec(\bV_B^{(k)})$ to $\bV_B^{(k)}$
      \ENDFOR
     \FOR{$l=1,..,L$}
      \STATE Compute the residual matrix $\bX_{\cdot}^{(l)} = \bX_{\cdot} - \sum^{K}_{k=1} \bU_B^{(k)}\bV_B^{(k)T} \bY_{\cdot}^{(k)} - \sum^{L}_{l'=1,l'\neq l}  \bU_S^{(l')}\bV_S^{(l')T}$
      \STATE Set $\bX_{j}^{(l)} = \boldmath{0}$ where $\bC_s[j,l]=0$ for $j=1,...,J$
    \STATE Update $\bU_S^{(l)} = \bX_{\cdot}^{(l)}\bV_S^{(l)}(\bV_S^{(l)T}\bV_S^{(l)} + \lambda_S^{(l)}\bI_{r_S})^{-1}$
      \STATE Update $\bV_S^{(l)} = \bX_{\cdot}^{(l)T}\bU_S^{(l)}(\bU_S^{(l)T}\bU_S^{(l)} + \lambda_S^{(l)}\bI_{r_S})^{-1}$
     \ENDFOR
    \ENDWHILE
     \STATE Set $\bB_k=\bU_B^{(k)}\bV_B^{(k)T}$ for all $k=1,..,K$, and $\bS_{\cdot}^{(l)}=\bU_S^{(l)}\bV_S^{(l)T}$ for all $l=1,..,L$
  \end{algorithmic}
\end{algorithm}

The symbol $\bigotimes$ means Kronecker product.

 Note that Algorithm 1 does not require the columns of $\bY_{\cdot}^{(k)}$ to be orthogonal. When $\bY_{\cdot}^{(k)}$ is semi-orthogonal, in light of \cref{nuclearA} and \cref{nuclearAY}, we develop an alternative approach based on iterative soft-singular value thresholding estimators for \eqref{ObjZ} in Algorithm 2.

\begin{algorithm}[H]
    \caption{Alternating Least Square with Soft-threshold Estimators} \label{ALSBS}
  \begin{algorithmic}[1]
    \INPUT Orthogonal covariates $\bY$ and corresponding multivariate outcomes $\bX$; penalizing terms $\lambda_B, \lambda_S$; binary indicator matrices $\bC_Y, \bC_S$
    \OUTPUT $\bB,\bS$
    \STATE \textbf{Initialization} Construct $\{\bY_{\cdot}^{(k)}\}_{k=1}^K$ based on $\bC_Y$. Assign initialized numbers for each entry of $\{\bB_k\}^K_{k=1},\{\bS_{\cdot}^{(l)}\}^L_{l=1}$
    \WHILE{convergence criterion does not meet}
    \FOR{$k=1,..,K$}
    \STATE Compute the residual matrix $\bX_{\cdot}^{(k)} = \bX_{\cdot} - \sum^{K}_{k'=1,k'\neq k} \bB_{k'} \bY_{\cdot}^{(k')} - \sum^{L}_{l=1}  \bS_{\cdot}^{(l)}$
    \STATE Compute the SVD of $\bX_{\cdot}^{(k)}\bY_{\cdot}^{(k)T}$, i.e. $\bX_{\cdot}^{(k)}\bY_{\cdot}^{(k)T} = \bL_B^{(k)}\bD_B^{(k)}\bR_{B}^{(k)}$
    \STATE Update $\bB_{k} = \bL_{B}^{(k)}\widehat{\bD}_B^{(k)}\bR_{B}^{(k)}$ where $\widehat{\bD}_B^{(k)}$ is a diagonal matrix with $\widehat{\bD}_B^{(k)}[r,r] = max(\bD_B^{(k)}[r,r]-\lambda_B^{(k)},0)$ for $r=1,2,...$ on its diagonal entries and zero otherwise
    \ENDFOR
    \FOR{$l=1,..,L$}
    \STATE Compute the residual matrix $\bX_{\cdot}^{(l)} = \bX_{\cdot} - \sum^{K}_{k=1} \bB_{k'} \bY_{\cdot}^{(k')} - \sum^{L}_{l=1,l'\neq l}  \bS_{\cdot}^{(l')}$
    \STATE Set $\bX_{j}^{(l)} = \boldmath{0}$ where $\bC_s[j,l]=0$ for $j=1,...,J$
    \STATE Compute the SVD of $\bX_{\cdot}^{(l)}$, i.e. $\bX_{\cdot}^{(l)} = \bL_{S}^{(l)}\bD_S^{(l)}\bR_{S}^{(l)}$
    \STATE Update $\bS_{\cdot}^{(l)} = \bL_{S}^{(l)}\widehat{\bD}_S^{(l)}\bR_{S}^{(l)}$ where $\widehat{\bD}_S^{(l)}$ is a diagonal matrix with $\widehat{\bD}_S^{(l)}[r,r] = max(\bD_S^{(l)}[r,r]-\lambda_S^{(l)},0)$ for $r=1,2,...$ on its diagonal entries and zero otherwise
    \ENDFOR  
    \ENDWHILE
  \end{algorithmic}
\end{algorithm}

 For both algorithms, we use the same convergence criteria to decide whether to stop the optimization process:
    $\sum^{K}_{k=1} ||\widehat{\bB}_k - \widetilde{\bB}_k ||_F^2 + \sum^{L}_{l=1} || \widehat{\bS}_{\cdot}^{(l)} - \widetilde{\bS}_{\cdot}^{(l)}||_F^2 < \epsilon$,
where $\widehat{}$ denotes the estimation in the current epoch and \, $\widetilde{}$ \,  denotes the estimation in the previous epoch. It is also reasonable to use convergence of the loss function as the criteria.

%From simulation results,  \cref{ALSUV} is more steady than  \cref{ALSBS} from three aspects:

%\begin{enumerate}
 %   \item When the penalties for $\bS$ is not at its theoretical best as described before, mse of $\bS$ for Algorithm 1 is smaller than that of Algorithm 2. The update steps for the latter involves soft-thresholding, which is very sensitive to the choice of penalties.
   % \item With best theoretical penalties for all estimates, Algorithm 1 and 2 require similar computation time to achieve the same convergence criterion. Generally, Algorithm 1 needs relatively less computation time. 
   % \item Theoretically, Algorithm 2 can be used only when we orthogonalize the original $\bY$, or the loss function will not constantly goes down, let alone converge.
   % \item In the case that the true rank for structures are large, algorithm 2 will give a quicker solution. In order to achieve the same performance as algorithm 2, we need to set the upper bound for rank estimate to be sufficient large, which takes too many computational sources. 
%\end{enumerate}

In practice both algorithms have different strengths and weaknesses. Theoretically, Algorithm 2 can be used only when we orthogonalize the original $\bY$, because otherwise soft-thresholding to update $\bB$ is not possible.  In general, we find that the algorithms require similar computation time to achieve the same convergence criterion: Algorithm 1 tends to require less time if the true ranks (and accompanying maximum ranks specified, i.e. $r_{B,upper}$ and $r_{S,upper}$) are small, while Algorithm 2 is quicker and consumes less computational resources when the true rank and maximum ranks specific for Algorithm 1 are large.

%\vspace{-.8cm}

\subsection{Missing data imputation}

One of the main uses of our proposed method is to impute various types of missing data. Based on the assumption that the abundance of existing entries provides sufficient information to uncover the global structures (both covariate and auxiliary effects) and therefore, to estimate the %absence of certain entries. 
values of absent entries.  Denote the set of indexes of all missing entries as $\textit{M}$. Our iterative imputation process is as follows: (1) Initialize $\widetilde{\bX}_{\cdot}$ by 
% $ \widetilde{\bX}_{\cdot} [m,n] = 
%   \begin{dcases*} 
%   $0$ & if  $[m,n] \in \textit{M}$ \\ 
%   \bX_{\cdot}[m,n] & if  $[m,n] \notin \textit{M}$.
%   \end{dcases*}$
$ \widetilde{\bX}_{\cdot} [m,n] = 
  \bX_{\cdot}[m,n] $ if  $[m,n] \notin \textit{M}$, otherwise 0;
  (2) Estimate $\{\bB_k\}^K_{k=1},\{\bS_{\cdot}^{(l)}\}^L_{l=1}$ by Algorithm 1 or 2 with current $\widetilde{\bX}_{\cdot}$;
  (3) Update $\widetilde{\bX}_{\cdot}$ by setting $\widetilde{\bX}_{\cdot} [m,n] = (\sum^{K}_{k=1} \bB_k \bY_{\cdot}^{(k)} + \sum^{L}_{l=1}  \bS_{\cdot}^{(l)})[m,n]$ for all $[m,n] \in \textit{M}$;
  (4) Back to (2) unless convergence; the final $\widetilde{\bX}_{\cdot}$ is the imputation result.
This can be considered a modified EM-algorithm, and is similar to the approach used for softImpute \citep{mazumder2010spectral} for nuclear-norm penalized imputation of a single matrix with no covariates.

%\vspace{-1cm}

\section{Simulations}
%%\vspace{-0.2cm}

\subsection{Recovery of true structure for special cases}

Here, we present simulations as proof-of-concept for two novel scenarios within our approach: (i) simultaneous modeling of covariate effects and auxiliary low-rank variation and (ii) simultaneous modeling of shared or specific covariate effects across multiple cohorts.  

For (i), we consider a single data matrix $\bX: 100 \times 100$ and single set of covariates $\bY: 10 \times 100$ and generate data via $\bX=\bB \bY+\bS+\bE,$ where $\bB \bY$ is covariate-driven variation, $\bS$ is auxiliary structured variation, and $\bE$ is error.  The coefficient array $\bB$ has rank $R_y$ via $\bB=a \bU_B \bV_B^T$ where $\bU_B: 100 \times R_y$ and $\bV_B: 10 \times R_y$, and $\bS$ has rank $5$ via $\bS=b \bU_S \bV_S$ where $\bU_S: 100 \times 5$ and $\bV_S: 5 \times 100$. The entries of $\bE$, $\bY$, $\bU_B$, $\bV_B$, $\bU_S$ and $\bV_S$ are all generated independently from a standard normal distribution. We consider $R_y=1$ or $R_y=5$, and consider three conditions with different signal strength for each term by adjusting $a$ and $b$: sd$(\bB \bY)=0.5$ and sd$(\bS)=5$ ($||\bB \bY||/||\bS||=0.1$,    sd$(\bB \bY)=\text{sd}(\bS)=1$ ($||\bB \bY||/||\bS||=1$), and sd$(\bB \bY)=5$ and sd$(\bS)=0.5$ ($||\bB \bY||/||\bS||=10$).  For each set of conditions, we estimate $\bB$ and $\bS$ using four approaches: 
(1) Augmented reduced rank regression (aRRR), our proposed approach as described in \cref{sec:obejective}, \textcolor{black}{given 10 as the rank upper bound for $\bB$ and $\bS$}.
(2) Supervised singular value decomposition (SupSVD) \citep{LI2016supSVD}, a related model of the form  $\bX = \bY\bB\bV^T + \bF\bV^T + \bE$ for one cohort, where $\bF$ is the matrix of latent variables that correspond to auxiliary variation not related to the covariates, estimated using maximum likelihood and given the true rank of $\bB\bV^T$ and $\bF\bV^T$.
(3) Two-stage least squares, in which the coefficients $\bB$ is determined by ordinary least squares regression and $\bS$ is determined by an SVD approximation with the true rank ($R=5$) on the residuals $\bX-\hat{\bB} \bY$.
(4) Two-stage nuclear norm (NN), in which $\bB$ is determined by an NN-penalized reduced rank regression and $\bS$ by a NN-penalized matrix approximation to the residuals $\bX-\hat{\bB} \bY$. For each method, we compute the relative mean squared error (MSE) for $\bB$ and $\bS$, e.g., $||\bB-\hat{\bB}||_F^2/||\bB||_F^2$. Average relative MSEs for each condition, over $100$ replications, are shown in Table~\ref{tab: sim_results}A.  This demonstrates clear advantages of a nuclear norm penalty on $\bB$, and the dramatic advantage of aRRR when the auxiliary signal $\bS$ is strong.  The latter point is critical, because molecular data typically have a large amount of structured variation that is driven by coordinated biological processes or other latent effects; it is common for such variation to be stronger than the signal of interest (i.e., $\bB \bY$), yet it is not systematically adjusted for in practice. 
   
For scenario (ii), we generate data $\{\bX_j: 100 \times 100, \bY_j: 10 \times 100\}$ via $\bX_j=(\bB+\bB_j)\bY+\bE$ for two cohorts $j \in \{1, 2\}$.  Here, $\bB_j$ are covariate effects specific to cohort $j$ and $\bB$ are shared effects. The coefficient arrays are generated via $\bB=a \bU_B \bV_B^T$,  $\bB_1=b \bU_{B_1} \bV_{B_1}^T$, and  $\bB_2=b \bU_{B_2} \bV_{B_2}^T$ where $\{\bU_B, \bU_{B_1}, \bU_{B_2}\}$ are each $100 \times R_y$ and $\{\bV, \bV_{B_1}, \bV_{B_2}\}$ are each $R_y \times 10$.  The entries of $\{\bE, \bY, \bU_B, \bU_{B_1}, \bU_{B_2}, \bV_B, \bV_{B_1}, \bV_{B_2}\}$ are each generated independently from a standard normal distribution. We consider $R_y=1$ or $5$, and three conditions with different signal strength for each term by adjusting $a$ and $b$: $a=2$ and $b=0.2$ $(||\bB||/||\bB_i||=10)$, $a=b=1$ ($(||\bB||/||\bB_i||=1)$, and $a=0.2$ and $b=2$ ($||\bB||/||\bB_i||=0.1$).   For each set of conditions, we estimate $\bB$, $\bB_1$ and $\bB_2$ for $J=2$ via maRRR with no auxiliary terms $\bS$, termed multi-cohort reduced rank regression (mRRR). Table~\ref{tab: sim_results}B shows average relative MSEs of $\bB$ and the $\bB_i$'s for mRRR in comparison to two-stage approaches analogous to those described previously.    The mRRR approach can effectively recover shared and cohort specific effects, with dramatic improvement over ad-hoc multi-step approaches.  

 \begin{table}[H]
 \caption{Relative MSE for scenarios assessing aRRR (\textbf{A}) and mRRR (\textbf{B}). Values that smaller than 0.01 are round to 0.01. The bold number represents the lowest value in a row.}	
\label{tab: sim_results}

\centering
\begin{tabular}{|r r|cc|cc|cc|cc|}
\hline
             \multicolumn{2}{|c|}{\textbf{A}} & \multicolumn{2}{|c|}{aRRR}                      & \multicolumn{2}{|c|}{SupSVD}                    & \multicolumn{2}{|c|}{Two-stage LS}              & \multicolumn{2}{|c|}{Two-stage NN}              \\ 
             \hline
$\frac{||\bB\bY||}{||\bS||}$ & $R_y$ & \multicolumn{1}{c}{$\bB$} & \multicolumn{1}{c|}{$\bS$} & \multicolumn{1}{c}{$\bB$} & \multicolumn{1}{c|}{$\bS$} & \multicolumn{1}{c}{$\bB$} & \multicolumn{1}{c|}{$\bS$} & \multicolumn{1}{c}{$\bB$} & \multicolumn{1}{c|}{$\bS$} \\ \hline
10           & 1    & $\bold{0.01} $                & 0.61                  & 0.01                  & $\bold{0.46}$                 & 0.01                  & 0.60                  & 0.01                  & 0.61                  \\
1            & 1    & $\bold{0.04}$                  & 0.22                  & 0.13                  & $\bold{0.19}$                  & 0.22                  & 0.21                  & 0.05                  & 0.25                  \\
0.1          & 1    & $\bold{0.17}$                  & $\bold{0.01}$                  & 10.97                 & 0.10                  & 11.44                 & 0.11                  & 7.45                  & 0.08                  \\
10           & 5    & $\bold{0.01}$                  & 0.63                  & 0.01                  & $\bold{0.48}$                  & 0.01                  & 0.60                  & 0.01                  & 0.63                  \\
1            & 5    & $\bold{0.14}$                  & 0.24                  & 0.17                  & $\bold{0.19}$                  & 0.23                  & 0.21                  & 0.15                  & 0.26                  \\
0.1          & 5    & $\bold{0.40}$                  & $\bold{0.01}$                  & 11.31                 & 0.10                  & 11.66                 & 0.11                  & 7.83                  & 0.08                  \\ \hline
\end{tabular}
%\vspace{20 pt}

\begin{tabular}{|r r|cc|cc|cc|}
\hline
	\multicolumn{2}{|c|}{\textbf{B}}& \multicolumn{2}{c|}{mRRR}&\multicolumn{2}{c|}{Two-stage LS}& \multicolumn{2}{c|}{Two-stage NN}\\ \hline
	$\frac{||\bB||}{||\bB_i||}$ & $R_y$ & $\bB$ & $\bB_i$ & $\bB$ & $\bB_i$ & $\bB$ & $\bB_i$\\
	\hline
	10 & 1 &$\bold{0.01}$ & $\bold{0.11}$ & 0.01 & 0.77 & 0.01 & 0.42 \\ 
	1& 1& $\bold{0.01}$ & $\bold{0.01}$ & 0.75 & 0.60 & 0.67 & 0.52  \\
	0.1& 1& $\bold{0.07}$ & $\bold{0.01}$ & 80.57 & 0.55 & 76.17 & 0.52 \\
	10 & 5 & $\bold{0.01}$ & $\bold{0.28}$ & 0.01 & 0.57 & 0.01 & 0.50\\
	1& 5& $\bold{0.08}$ & $\bold{0.08}$ & 0.49 & 0.54 & 0.45 & 0.50 \\
	0.1& 5& $\bold{0.49}$ & $\bold{0.01}$ & 54.79 & 0.56 & 52.41 & 0.54 \\
\hline
\end{tabular}
\end{table}

\subsection{Missing data imputation}
\label{missingsims}

In this simulation we assess the maRRR framework more broadly, with a focus on missing data imputation. Our general simulation procedure follows these steps:  1) complete data generation; 2) missingness assignment; 3) imputation analysis. In reality the true main signals may come from covariate effects or auxiliary structures and can be individual-level or shared across multiple cohorts. So we consider four fundamental scenarios: ($a$) large $\bB$, main signals from one global auxiliary structure which is shared by all cohorts; ($b$) large $\bS$, main signals from one global covariate effect which is shared by all cohorts; ($c$) large $\bB_i$, main signals from individual covariate effect of each cohort;  ($d$) large $\bS_i$, main signals from individual auxiliary structure in each individual cohort. In order to mimic the real situation, the number of samples and dimensions of the data is set to be the same as the TCGA data analyzed in Section~\ref{sec:data}. That is, $\bX$ consists of 1000 features and 6581 samples from 30 study cohorts and $\bY$ consists of 50 predictors. Therefore, the ground truth can be written as  $\bX_j = \bB\bY_j + \bB_j \bY_j + \bS_{shared,j} + \bS_j + \bE_j, j = 1,...,30$. In each simulation, the standard deviation for the main signals is set to be $\sqrt{10}$ while that of the remaining signals and random errors are set to be 1.  The complete data generation process is described in Appendix~\ref{app:data_gen}.

\begin{table}[]
\caption{Imputation relative squared error(RSE) under different methods and different types of missingness, simulated data is set to be large at only one type of modules. Missingness is set to be $5\%$ of the original $\bX$. The number of epochs for each method is set as 30. Each result is a mean of 10 replications. The standard error is less than 0.01 for all of the means shown. The bold number represent the lowest value in a column.}
\resizebox{\textwidth}{!}{
\begin{tabular}{llllll}
\hline
large\_B  & Method                    & missing entries & missing columns & missing rows & mean  \\ \hline
          & maRRR                     & $\bold{0.082}$           & $\bold{0.228}$           & $\bold{0.216}$        & $\bold{0.175}$ \\
          & BIDIFAC+                  & 0.085           & 1               & 0.231        & 0.439 \\
          & mRRR                      & 0.202           & 0.241           & 0.285        & 0.243 \\
          & aRRR, one all-shared      & 0.125           & 0.288           & 0.227        & 0.213 \\
          & aRRR, 30 separate         & 0.093           & 0.287           & 1.014        & 0.465 \\
          & NN reg, one all-shared    & 0.283           & 0.287           & 0.288        & 0.286 \\
          & NN reg, 30 separate       & 0.212           & 0.255           & 1            & 0.489 \\
          & NN approx, one all-shared & 0.127           & 1               & 0.229        & 0.452 \\
          & NN approx, 30 separate    & 0.096           & 1               & 1            & 0.699 \\ \hline
large\_S  & Method                    & missing entries & missing columns & missing rows & mean  \\ \hline
          & maRRR                     & $\bold{0.082}$           & $\bold{0.877}$           & $\bold{0.218}$        & $\bold{0.392}$ \\
          & BIDIFAC+                  & 0.085           & 1               & 0.225        & 0.437 \\
          & mRRR                      & 0.759           & 1.066           & 0.927        & 0.917 \\
          & aRRR, one all-shared      & 0.125           & 0.929           & 0.228        & 0.427 \\
          & aRRR, 30 separate         & 0.093           & 0.884           & 1.004        & 0.66  \\
          & NN reg, one all-shared    & 0.928           & 0.931           & 0.933        & 0.93  \\
          & NN reg, 30 separate       & 0.783           & 1.072           & 1            & 0.952 \\
          & NN approx, one all-shared & 0.127           & 1               & 0.23         & 0.452 \\
          & NN approx, 30 separate    & 0.096           & 1               & 1            & 0.699 \\ \hline
large\_Bi & Method                    & missing entries & missing columns & missing rows & mean  \\ \hline
          & maRRR                     & $\bold{0.083}$           & 0.262           & 0.901        & $\bold{0.415}$ \\
          & BIDIFAC+                  & 0.086           & 1               & $\bold{0.867}$        & 0.651 \\
          & mRRR                      & 0.204           & $\bold{0.246}$           & 0.928        & 0.459 \\
          & aRRR, one all-shared      & 0.149           & 0.913           & 0.902        & 0.655 \\
          & aRRR, 30 separate         & 0.093           & 0.287           & 1.013        & 0.464 \\
          & NN reg, one all-shared    & 0.896           & 0.899           & 0.964        & 0.92  \\
          & NN reg, 30 separate       & 0.212           & 0.252           & 1            & 0.488 \\
          & NN approx, one all-shared & 0.151           & 1               & 0.907        & 0.686 \\
          & NN approx, 30 separate    & 0.096           & 1               & 1            & 0.699 \\ \hline
large\_Si & Method                    & missing entries & missing columns & missing rows & mean  \\ \hline
          & maRRR                     & $\bold{0.083}$           & $\bold{0.873}$           & $\bold{0.861}$        & $\bold{0.606}$ \\
          & BIDIFAC+                  & 0.086           & 1               & 0.866        & 0.651 \\
          & mRRR                      & 0.76            & 1.066           & 0.928        & 0.918 \\
          & aRRR, one all-shared      & 0.148           & 0.93            & 0.906        & 0.661 \\
          & aRRR, 30 separate         & 0.093           & 0.885           & 1.003        & 0.661 \\
          & NN reg, one all-shared    & 0.929           & 0.931           & 0.935        & 0.932 \\
          & NN reg, 30 separate       & 0.784           & 1.072           & 1            & 0.952 \\
          & NN approx, one all-shared & 0.15            & 1               & 0.91         & 0.687 \\
          & NN approx, 30 separate    & 0.096           & 1               & 1            & 0.699 \\\hline
\end{tabular}
}

\label{tab:sim_missing}
\end{table}

To mimic the various types of missingness that are encountered in reality, we conduct simulations in which four types of missingness are considered: (i) missing entries, (ii) missing columns, (iii) missing rows, and (iv) a balanced mix of these three types of missingness as the average of the results of those three types. Missingness is set to be 5$\%$ of the original data for the assumption that adequate information is provided for revealing global structures. All missing indices are randomly selected. Denote $\widetilde{\bX}_{\cdot}$ as the estimate for true observation $\bX_{\cdot}$, based on non-missing entries. We define the relative squared error (RSE) for missing data imputation as
    $RSE = \frac{\sum_{(m,n)\in\textit{M}}(\bX_{\cdot}[(m,n)] - \widetilde{\bX}_{\cdot}[(m,n)])^2}{\sum_{(m,n)\in\textit{M}}(\bX_{\cdot}[(m,n)] )^2}$.

We compare our method (maRRR with true 31 modules) with the following approaches:
    (1) BIDIFAC+ with 31 modules, i.e. only auxiliary variation structure $\bS$;
    (2) mRRR with 31 modules, i.e. only covariate-related structure $\bB \bY$;
    (3) aRRR with only one module for all cancer types' cohorts together;
    (4) aRRR separately on each cancer type's cohort;
    (5) nuclear norm regression (without $\bS$) of $\bX_{\cdot}$ on $\bY_{\cdot}$, i.e. mRRR with one all-shared module;
    (6) nuclear norm regression (without $\bS$) of $\bX_{j}$ on $\bY_{j}$ separate for each cancer type, i.e. mRRR with 30 individual modules;
    (7) nuclear norm approximation (without $\bB\bY$) for all cancer types together
    (8) nuclear norm approximation (without $\bB\bY$) for each cancer types separately.

%%\vspace{-.1in}
%The upper bound of rank for all estimation procedures are set as 20 but 50 for method (9) in the case of missing rows.
Based on the simulation results shown in Table~\ref{tab:sim_missing}, in terms of the average performance, our proposed method maRRR has the lowest RSE. In particular, maRRR has a very close RSE to the best result in the case of missing columns or rows under large individual covariate signals, and it performs the best in all other cases. BIDIFAC+ performs slightly worse than our proposed method while there are only missing entries or rows, but it cannot utilize any covariate information to impute in the case of missing columns. The models only considering covariate effects (mRRR and nuclear norm regression) cannot predict accurately when the auxiliary variation is large, no matter global or individual. On the contrary, the models without considering covariates (BIDIFAC+, nuclear norm approximation) can perform reasonably well in the cases of missing entries or rows since the variation from covariates may be counted into that of auxiliary structures. The special case of our proposed method, aRRR (for one cohort only), though worse than maRRR, has lower MSE than many other existing methods. %Thus, we can implement aRRR for missing imputation for faster usage, as long as the requirement for precision is not extremely strict and decomposing shared or individual structures across cohorts is not of interest.  %This sheds lights on various real data analyses beyond the data we discuss here.

%\vspace{-0.7cm}

\subsection{Computation}
Our proposed method is computationally efficient. For a matrix of size 1000 $\times$ 6581, the average computational cost per epoch is 45 seconds for Algorithm 1 and 50 seconds for Algorithm 2. A comprehensive comparison of computation times for all methods utilized in this study is provided in Appendix~\ref{app:computation}.
%Algorithm 2, loss converges after 20 epochs, 665508456 116946104  94741541  74391956  56391410  41183284  29006765  19826518  13485311 9720814   8094344   7724408   7661098   7637280   7626412   7621401   7618943   7617568 7616738   7616212. Including calculating the loss, it takes 9600 seconds for 20 epochs, which means 480 seconds per epoch. It means the loss calculating is time taking. Algorithm 1 takes 469 seconds for one epoch. Compared with previous results, it means loss calculation takes 420 seconds on average per epoch. The loss looks like 105195141  14382363   9851606   9358367   9152919   8981609   8825912   8681309   8544952  8415842   8293905   8180008   8075553   7981876   7900239   7831741   7776439   7733369  7700596   7675959, which converges slower than algorithm 2 when it is near optimal.

%On msi with the following configuration: #SBATCH --nodes=1#SBATCH --ntasks-per-node=1#SBATCH --mem-per-cpu=32g,it takes 7 minutes to run one epoch on all 9 methods.

%\vspace{-1cm}

\section{Real Data Analysis}
\label{sec:data}
%\vspace{-0.2cm}
\subsection{Data description}

We consider data from the \textcolor{black}{Cancer Genome Atlas (TCGA)} Pan-Cancer Project \citep{hoadley2018cell}. \textcolor{black}{TCGA is an NIH-sponsored initiative to molecularly characterize cancer tissue samples obtained from hundreds of sites worldwide.}     We used data for 6581 tumor samples \textcolor{black}{from distinct individuals} representing 30 different cancer types (i.e., 30 cohorts). \textcolor{black}{The number of samples per cancer type ranges from 57 samples for uterine carcinosarcoma (UCS) to 976 for breast carcinoma (BRCA).} As outcomes, we consider \textcolor{black}{gene expression data obtained from Illumina RNASeq platforms and normalized as described in \citet{hoadley2018cell}}. We filter to the 1000 genes that have the highest standard deviation, yielding $\bX_{\cdot}: 1000 \times 6581$. We filter to the 50 somatic most common somatic mutations (1=mutated and 0=not mutated)  as covariates $\bY_{\cdot}: 50 \times 6581$.  Data are standardized and orthogonalized as in Appendix~\ref{app:scaling}.  Further details on the data are available in Appendix~\ref{app:names} and \ref{app:data_process}.

%\vspace{-0.5cm}
\subsection{Decomposition results}
\label{sec:decomposition}
%Denote $\widehat{\bX}_{\cdot}$ as the estimate for true observation $\bX_{\cdot}$. We define the relative squared error (RSE) as the measurement for accuracy:
%\begin{align*}
%    RSE = \frac{||\bX_{\cdot} - \widehat{\bX}_{\cdot}||_F^2}{||\bX_{\cdot}||_F^2}
%\end{align*}

We first apply the optimization with dynamic modules for BIDIFAC+ \citep{Lock2022BADIFACplus}, to uncover 50 low-rank modules in $\bX_{\cdot}$. \textcolor{black}{This stepwise procedure iteratively determines modules of shared structure without consideration of any covariates (i.e., $\bC_S$).} %It sheds light on the construction of modules that explains the most variance.}  
Fifty was chosen as the upper bound because the variance explained by more modules than 50 was relatively inconsequential.  To distinguish how much variance is attributed to mutation effects, we set covariate-related module indicators $\bC_Y$ equal to $\bC_S$. We then apply maRRR to these $50$ modules to infer both mutation-driven and auxiliary structured variation, ($K=L=50$)  with penalties determined as in Section~\ref{sec:theory}.  % With the detected 50 modules, maRRR achieves a RSE of 0.184, which indicates considerable amounts of variation is captured. The rank upper bound for all estimates is set to be 20 since it is a reasonable cut-off for low rank. In reality, the true ranks of most of the estimates are around 10.

\begin{sidewaystable}[]
\caption{Cancer types and sources for the first 15 modules, ordered by variation explained by maRRR. Variance of $\bS_{\cdot}^{(i)}$/$\bB_i\bY^{(i)}_{\cdot}$/signal refers to total variance explained by $\bS_{\cdot}^{(i)}$/$\bB_i\bY^{(i)}_{\cdot}$/$\bB_i\bY^{(i)}_{\cdot}+\bS_{\cdot}^{(i)}$ estimated in maRRR. RSE of mRRR refers to relative square difference between $\bB_i\bY^{(i)}_{\cdot}$ estimated by maRRR and mRRR. RSE of BIDIFAC+ refers to relative square difference between signal ($\bB_i\bY^{(i)}_{\cdot}+\bS_{\cdot}^{(i)}$) estimated by maRRR and BIDIFAC+. Sample size refers to the number of samples used in the current module. All the number smaller than 0.01 are round to 0.01. For reference, the RSE for mRRR and maRRR is defined as $\frac{||\bB_{i,mRRR} - \bB_{i,maRRR}||_F^2}{||\bB_{i,mRRR} + \bB_{i,maRRR}||_F^2}$ $ = \frac{||\bB_{i,mRRR}\bY_{\cdot}^{(i)}-\bB_{i,maRRR}\bY_{\cdot}^{(i)}||_F^2}{||\bB_{i,mRRR}\bY_{\cdot}^{(i)} + \bB_{i,maRRR}\bY_{\cdot}^{(i)}||_F^2}$ since $\bY_{\cdot}^{(i)}$ is semi-orthogonal; the RSE for mRRR and maRRR is defined as $\frac{||\bS_{\cdot,BIDIFAC+}^{(i)} - \bS_{\cdot,maRRR}^{(i)}||_F^2}{||\bS_{\cdot,BIDIFAC+}^{(i)} + \bS_{\cdot,maRRR}^{(i)}||_F^2}$ for any $i =1,2,...,50$.}
\label{tab: top15mods}
\begin{tabular}{cccccccc}
\hline
\hline
Module & Sample  & Variance  & RSE of & Variance  & Variance & RSE of & Cancer  \\
($i$) & size & of $\bB_i\bY^{(i)}_{\cdot}$ &  mRRR & of $\bS_{\cdot}^{(i)}$ & of signal &  BIDIFAC+ & types \\ \hline                     \\
1      & 976         & 129652.09     & 0.22         & 1085292.83    & 1524117.23         & 0.01           & BRCA                               \\
2      & 400         & 167384.83     & 0.15         & 551269.19     & 992693.82          & 0.01           & THCA                               \\
3      & 433         & 67242.88      & 0.20         & 645839.92     & 975136.09          & 0.01           & GBM, LGG                           \\
4      & 331         & 49.52         & 0.97         & 730462.99     & 734935.73          & 0.01           & PRAD                               \\
5      & 195         & 7231.61       & 0.68         & 633034.84     & 724960.62          & 0.01           & LIHC                               \\
6      & 1029        & 137934.04     & 0.08         & 316746.84     & 656053.3           & 0.03           & BLCA, CESC, ESCA, HNSC, KICH, LUSC \\
7      & 820         & 137765.92     & 0.14         & 298195.14     & 629007.34          & 0.02           & COAD, ESCA, PAAD, READ, STAD       \\
8      & 179         & 20.61         & 1.00         & 396671.22     & 396882.04          & 0.01           & PCPG                               \\
9      & 170         & 53.04         & 0.98         & 347389.29     & 348257.22          & 0.01           & LAML                               \\
10     & 576         & 13530.5       & 0.33         & 231107.59     & 316968.61          & 0.04           & KIRC, KIRP                         \\
11     & 422         & 30138.96      & 0.21         & 192612.5      & 300435.91          & 0.03           & SKCM, UVM                          \\
12     & 6411        & 0.14          & 0.74         & 264645.95     & 264654.13          & 0.23           & All but *LAML*                     \\
13     & 211         & 63836.59      & 0.16         & 87452.29      & 251028.58          & 0.01           & COAD, READ                         \\
14     & 149         & 495.86        & 0.87         & 237614.82     & 250811.31          & 0.01           & TGCT                               \\
15     & 119         & 97.81         & 0.97         & 242978.25     & 243471.26          & 0.01           & THYM                            \\\hline  
\end{tabular}
\end{sidewaystable}

We order the modules by total variance explained by descending by maRRR estimates, i.e. \textcolor{black}{$||\widehat{\bB}_i \bY_{\cdot}^{(i)} + \widehat{\bS}_{\cdot}^{(i)}||_2^2, i = 1,...,50.$} The ordered result is shown in Table~\ref{tab: top15mods}. %Top modules have a variety of combinations of cancer types. Both shared modules and individual modules address adequate variance. The variance explained by each module differs notably. This demonstrates a must of allowance for multiple flexible architectures. 
The top 3 modules by total variance explained are those with one or two cancers: BRCA, THCA and a combination of GBM and LGG (both neurological cancers), respectively. In general, auxiliary structures explained more variation than mutation-related structures, but their relative contribution varied widely across modules.  For example, Modules 6 and 7 have fairly comparable amount of variation explained by both mutation-related and -unrelated parts. Other modules have negligible mutation-driven variation, such as Module 12 which is shared by all but LAML. The large amount of low-rank variation unrelated to covariates demonstrates the importance of accounting for this auxiliary structure.  Moreover, the large amount of  covariate-related and -unrelated variation that is specific to one or a small number of cancer types demonstrates the importance of accounting for individual and partially-shared structures.  
%and reveals why regression-based methods do not achieve a satisfactory performance when data has the low-rank property.  In all, for each module, variance explained by $\bB \bY$ and $\bS$ varies a lot so we need to include both parts to detect different unique signals. Examples are given below.

\textcolor{black}{ We present a comparison of the BIDIFAC+ and mRRR estimates with those of our proposed maRRR in Table \ref{tab: top15mods}.
The total signal detected by maRRR closely aligns with that of BIDIFAC+, illustrating maRRR's ability to discern covariate effects while accounting for similar total variance. The covariate effects identified by mRRR resemble those by maRRR, particularly when the sample size is large. However, mRRR tends to estimate larger covariate effects, especially for smaller sample sizes.  This makes sense because the maRRR estimates are less prone to over-fitting by adjusting for unrelated structure. As the sample size decreases, the relative square errors (RSE) between estimates from the two methods increase. However, both methods generally agree when the mutation effects are negligible ($\leq$ 10000). For instance, both mRRR and maRRR reveal virtually no global mutation effects for Module 12 (even though the RSE is large due to the relative standardization).}   %Larger numbers of RSE are also observed when the variance of the covariate effects is already negligible ($\leq$ 10000), causing both methods to estimate the covariate effects as insignificant. For instance, both mRRR and maRRR reveal virtually no global covariate effects for Module 12 even RSE is large. Overall, the detected covariate structures from both methods largely concur. }
\begin{figure}[]
        \centering
        \includegraphics[height=17cm,width=15cm]{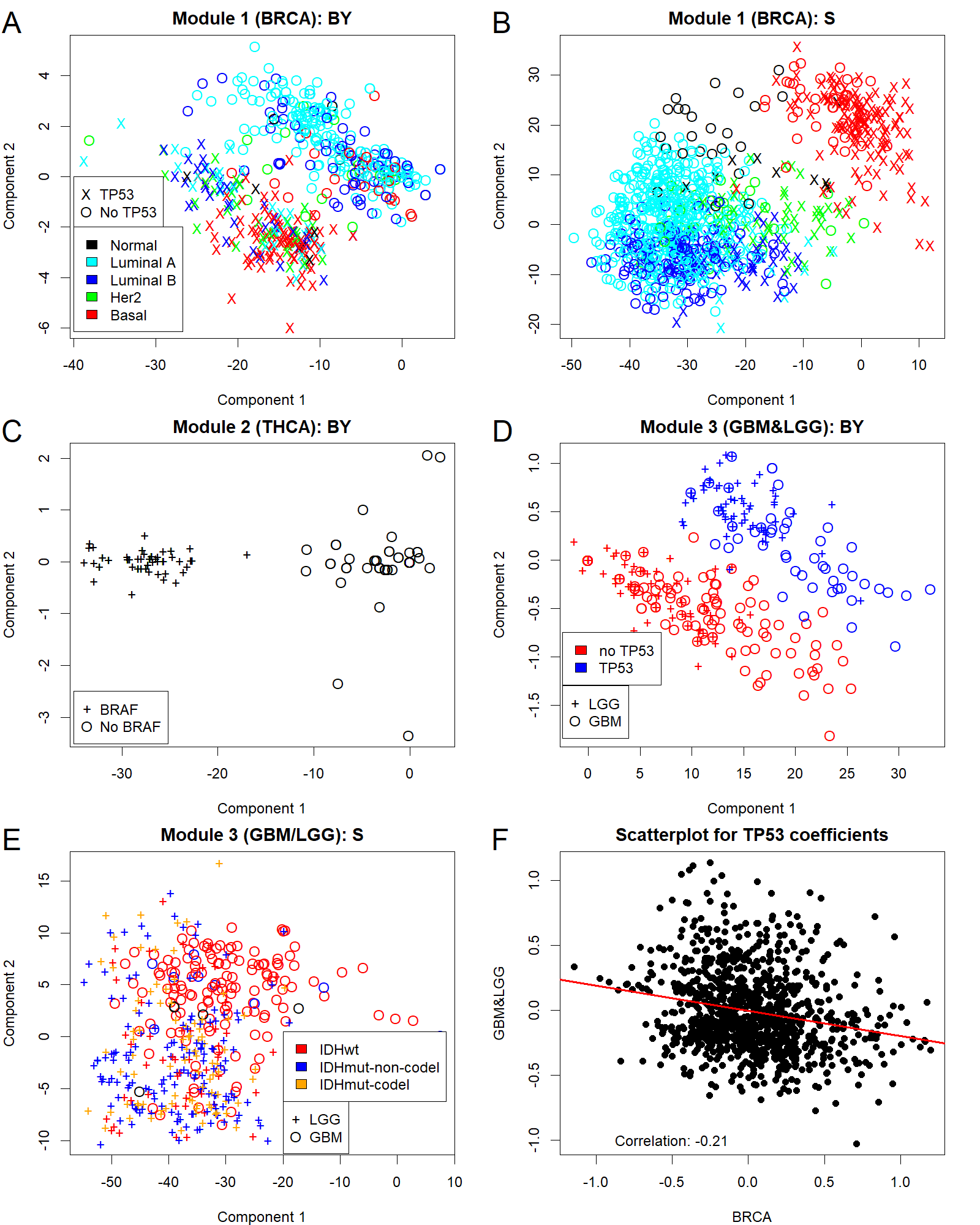}
        \caption{$\bA$: scores for the first two principle components of covariate-related variation ($\bB \bY)$ from Module 1 (BRCA); $\bB$: scores for the first two principle components of covariate-unrelated auxiliary variation ($\bS$) from Module 1 (BRCA), with symbols and colors showing TP53 mutation and 5 subtypes of BRCA respectively. $\bC$: scores for the first two principle components of covariate-related variation ($\bB \bY)$ from Module 2 (THCA), with symbols showing BRAF mutation. \textcolor{black}{Plot $\bA\&\bB$ contain 976 samples in the BRCA cohort respectively; Plot $\bC$ contains 400 samples in the BRCA cohort respectively.} $\bD$: scores for the first two principle components of covariate-related variation ($\bB \bY)$ from Module 3 (GBM\&LGG); $\bE$: scores for the first two principle components of covariate-unrelated auxiliary variation ($\bS$) from Module 3 (GBM\&LGG), with symbols and colors showing TP53 mutation and 3 IDH/codel subtypes of GBM\&LGG respectively. $\bF$: scatterplot for the regression coefficients for TP53 in module 1 and module 3. \textcolor{black}{Each point in Plot $\bD\&\bE$ represents one sample (150 for GBM and 283 for LGG); each point in Plot $\bF$ represents one gene (1000 in total).} }
        \label{fig:real_Data}
\end{figure}

Principal components plots of the first three modules (BRCA, THCA, GBM and LGG) %explain the most variance among 50 modules, whose plots of principle components 
are shown in Figure~\ref{fig:real_Data}. %BRCA is short for Breast invasive carcinoma. 
%From Figure 1.A, %there is a clear boundary between two shapes of points (whether containing TP53, which is one of 50 somatic mutations in the analysis).
Figure~\ref{fig:real_Data}A displays mutation-related variation for the BRCA module, and samples are distinguished by whether they have a mutation in the TP53 gene or not. This makes sense, as TP53 is known to play a critical role in genomic activity for cancer and it is the  most frequently mutated gene in breast cancer \citep{olivier2010tp53}. %\citep{Cancer2015TP53}.
%Emerging research \citep{Cancer2015TP53} discovers that TP53 is the most frequently mutated gene in breast cancer. 
%This demonstrates that the estimated covariate-related variation of BRCA captures the main signal from somatic mutations. 
From Figure~\ref{fig:real_Data}B,% we can see points of the same color are clustered together, though with some overlapping. 
we see that the mutation-unrelated structure is driven by the 5 intrinsic BRCA subtypes %\citep{Ciriello2015brca}:
\citep{Cancer2012brca}:
%This indicates that the auxiliary variation catches the 5 subtypes' information correctly. This agree with the well-known scientific finding that Breast invasive carcinoma (BRCA) are divided into 5 intrinsic subtypes:  
Normal-like, Luminal A (LumA), Luminal B (LumB), HER2-enriched (HER2), and Basal-enriched (Basal) tumors. %It would be straightforward to make classifications of in terms of the subtypes based on the first two principle components of estimated auxiliary structure $\bS$. 
This makes sense, as the BRCA subtypes are known to be genomically distinct, but (as is apparent) do not have a direct correspondence to TP53 or other common mutations. 
From Figure~\ref{fig:real_Data}C, we observe that BRAF and non-BRAF groups are elegantly separated on the first principle component of mutation-driven variation for THCA, which explains much more variation than other components. This concurs with prior research indicating that the BRAF mutation defines a unique genomic and clinical subgroup within THCA patients \citep{Dolezal2021BRAF}.
 
 Figure~\ref{fig:real_Data}D\&E reveal substantial variation in both the GBM and LGG samples, as evidenced by the spread and intermingling of the ``+" %(representing GBM samples' discriminant scores) 
and ``o" symbols. 
%(representing LGG samples' discriminant scores) 
 This observation indicates that the top components identified by Module 3 explain substantial variation in both the GBM and LGG cancer cohorts, suggesting that it is indeed shared by the two cancer types.
The TP53 mutation drives separation of the mutation-driven structure, which makes sense as TP53 status is closely related to GBM and LGG aggressiveness \citep{ham2019tp53}. This was also observed in BRCA, however, from Figure~\ref{fig:real_Data}F we see that the TP53 regression coefficients from module 3 (GBM\&LGG) has little correlation to those from module 1 (BRCA). This demonstrates that while TP53 is an important somatic mutation related to different types of cancer, its effect on gene expression can differ dramatically  depending on the cancer type. In general, we find that the same somatic mutation plays different roles for different cohorts and modules. This embodies the necessity of the flexible modeling for different (combinations of) cohorts.  In fact, one interesting finding of our analysis is that the effect of somatic mutations on gene expression are almost not entirely shared across different types of cancer.  For example, for the module that is shared across almost all cancers (module 12) the  mutation-driven component is negligible. This observation is further supported by our analysis in Appendix~\ref{app:pan_cancer} and Section~\ref{s:real_missing}.

%\vspace{-0.5cm}
\subsection{Missing data imputation}
\label{s:real_missing}

Similar to Section~\ref{missingsims}, we compare our proposed maRRR with other relevant methods under four types of missingness for these data.  Beyond the aforementioned eight methods in Section~\ref{missingsims}, we added (9) linear least squares regression to predict $\bX$ from $\bY$ for all cancer types together and (10) linear least-squares regression for each cancer types separately. Note that maRRR, BIDIFAC+ and mRRR are based on the detected 50 modules. Results are provided in Table~\ref{tab:real_data_missing}.  %We start the analysis from three single types of missingness.
%We first describe the results for the three different types of missingness separately.  
In the scenario of missing entries,
both maRRR and BIDIFAC+ have the lowest RSE, with similar values; both methods allow for an efficient decomposition of joint and individual structures.  In the case of missing columns (some samples' entire outcomes are missing), %covariates play a decisive role for forecast so it is no wonder that all methods that only consider $\bS$ (BIDIFAC+, NN approx) have no prediction power.
the methods that do not consider mutations and only consider $\bS$ (BIDIFAC+, NN approx) have no predictive power, which is expected because the mutation data is needed to inform predictions if no gene expression is available for a sample.
Here the methods that consider both $\bB$ and $\bS$ (maRRR, aRRR) are suboptimal compared to those methods that only consider mutation-driven variation ($\bB\bY$), indicating that for these data allowing for auxiliary variation does not improve column-wise predictions.   %Samples tend to share similarities within one cohort.  
Moreover, methods that allow for separate mutation-driven structure across the cohorts perform substantially better, which is consistent with the fact that there were more individual modules in our analysis and mutation-driven variation was generally not shared (Table~\ref{tab: top15mods}). %This leads that methods considering more individual structures (mRRR and NN reg for 30 seperate) win a better predictio
In the event of missing rows (each cohort misses random features), methods that consider individual structure only do not perform well, as they cannot leverage shared structure when a gene is entirely missing within a cohort.
%have no aid to forecast so aRRR and NN approx with 30 separate are incapable at all. 
On the contrary, methods with only one all-shared module (NN approx and aRRR) perform well. Here, maRRR also performs reasonable well, as including several individual modules does not limit %maRRR to achieve a reasonable
its performance. In this case methods considering covariate effects only (mRRR, NN reg)  do not perform well, as they tend toward estimates of zero (i.e., no predictions) to minimize squared error loss. %get estimates as zero to minimize square loss, i.e. no prediction as well. 
 NN approx is slightly better than aRRR., perhaps because it does not consider covariate-driven variation.

Under the circumstances of a balanced mix of missingness for different conditions, maRRR has the best average recovering ability. This is largely because it is the most robust and flexible. Other comparable methods (BIDIFAC+, aRRR, NN approx) will have limited peformnance for at least one form of missingness. In reality, maRRR will be the most suitable for imputation since missingness is unpredictable and complex.

\begin{table}[]
\caption{Imputation relative squared error(RSE) under different methods and different types of missingness. Missingness is set to be $5\%$ of the original $\bX$. ``one all shared" means data for 30 groups are stacked together to form one matrix to analyze; ``30 separate" means each group has its only model. ``Missing entries" refers to missingness is entrywise; ``missing columns" means some samples' entire observation are missing; ``missing rows" means each group has several features entirely missing. ``N/A" means some specific method is not applicable. The bold number represents the lowest value in a column.}
\begin{tabular}{|l|cccc|}
\hline
Methods                   & \multicolumn{1}{l}{Missing entries} & \multicolumn{1}{l}{Missing columns} & \multicolumn{1}{l}{Missing rows} & \multicolumn{1}{l|}{Average} \\ \hline
maRRR                     & $\bold{0.233}$                               & 0.813                               & 0.600                          & $\bold{0.548}$                        \\
BIDIFAC+                  & $\bold{0.233}$                              & 0.999                               & 0.613                            & 0.615                        \\
mRRR                      & 0.603                               & $\bold{0.711}$                               & 0.998                            & 0.770                        \\
aRRR, one all-shared      & 0.261                               & 0.930                               & 0.487                            & 0.559                        \\
aRRR, 30 separate         & 0.376                               & 0.780                               & 1.001                            & 0.719                        \\
LS reg, one all-shared    & 0.908                               & 0.906                               & 0.899                            & 0.904                        \\
LS reg, 30 separate       & 0.560                               & N/A                                 & N/A                              & N/A                          \\
NN reg, one all-shared    & 0.912                               & 0.913                               & 1.032                            & 0.953                        \\
NN reg, 30 separate       & 0.599                               & 0.727                               & 1.000                            & 0.775                        \\
NN approx, one all-shared & 0.273                               & 1.000                               & $\bold{0.454}$                            & 0.576                        \\
NN approx, 30 separate    & 0.252                               & 1.000                               & 1.009                            & 0.754                        \\ \hline
\end{tabular}
\label{tab:real_data_missing}
\end{table}

%\vspace{-1cm}

\section{Discussion}

%The efficient and accurate data integration of multi-cohort data is an attractive and challenging research topic for modern data analyses. Unlike other methods that only considers either matrix factorization or regression, we have proposed a supervised flexible matrix decomposition method, maRRR, for joint integrative analysis studies of multi-cohort data. Nuclear norm penalty and matrix factorization together impose the low-rank structures. The low-rank decomposition of both covariate effects and auxiliary structures, captures unique signals from multiple underlying combinations of cohorts. Aiming for efficient practical applications, we provide not only two optimization algorithms but also theoretical results such as best penalty selections and non-zero estimate premises. 
Two strengths of our proposed maRRR approach are its flexibility and versatility.  It is flexible because it accounts for various types of signals - covariate-driven, shared or unshared - without prior assumptions on the size or rank of these signals.   
It is versatile because it is capable of performing many tasks at once: e.g., dimension reduction, prediction and missing data imputation. These advantages are well-illustrated by our pan-cancer application, in which adequate amounts of variation are explained by different components and the patterns detected are both insightful and consistent with existing scientific research on cancer.

%As a model assumption, multi-cohort covariate effects are more intuitive and reasonable than multi-view ones as a result that all samples share common features, even from multiple cohorts.
We focus on multi-cohort integration rather than multi-view (data on the same subjects from different sources) integration, in part because shared or unshared covariate effects are straightforward to interpret across multiple cohorts.
But one can still argue that in a multi-view (e.g., multi-omics) context each sample will have intrinsic underlying signals that will affect variables from different sources. Without loss of generality, this method can be adapted to analyze multi-view data as well. This is achieved by simply switching the way we integrate matrices: horizontally across shared rows or vertically across shared columns. A promising future direction is to extend maRRR to the bidimensional integration context, where the data are both multi-cohort and multi-view. \textcolor{black}{Another direction of future work is alternative empirical approaches to determine the module indicator matrices $\bC_Y$ and $\bC_S$, such as via an iterative stepwise selection procedure extending that in \citep{Lock2022BADIFACplus}.} % This allows different cohort combinations in covariate-related effects and auxiliary structures as outputs rather than setting $\bC_Y = \bC_S$ to partition the sources of variance as we have done in Section 7.2.
While we have fixed the selection of penalty parameters $\lambda_B^{(k)},\lambda_S^{(l)}$ by employing random matrix theory, the parameters or the ranks of the underlying structures may be estimated empirically by a cross-validation procedure combined with a grid search.

Further theoretical developments are also a pertinent future direction, such as proving the convergence of our optimization algorithms to a global optimum and establishing sufficient conditions for the uniqueness of the solution. There is empirical evidence for both conjectures, as we find that the converged solution is the same with different initializations and for the two optimization algorithms considered. Moreover, Theorem 1 of \cite{Lock2022BADIFACplus} provides sufficient conditions for conditional uniqueness of the $\{\bS_{\cdot}^{(l)}\}^L_{l=1}$ given $\{\bB_k\}^K_{k=1}$, and vice-versa.  %In applications, examining plots for all other modules could lead to more meaningful scientific interpretations, as demonstrated in Figure~\ref{fig:real_Data}.
%}

%\vspace{-1cm}

\section*{Data Availability Statement}
The data that support the findings in this paper are provided via \href{https://www.dropbox.com/s/mkkuks54g62e9gy/30grps_50mods_est_0418_version2.RData?dl=0}{this RData file}.
 The user-friendly R package maRRR at \url{https://github.com/JiuzhouW/maRRR} performs all functions described herein, such as fitting models by the two algorithms in Section~\ref{sec:opt}, imputing missing values as in Section~\ref{missingsims}, generating penalties as in Section~\ref{sec:theory}, and generating data as in Appendix~\ref{app:data_gen}. For real data analysis, we provide all the model estimates as \href{https://www.dropbox.com/s/ub9zu5inxlbh6x5/30grps_50mods_est_0418_version2.RData?dl=0}{Rdata file} with detailed \href{https://www.dropbox.com/s/af743f8ocucz7p6/readme_modelFit.txt?dl=0}{notation explanations} and heatmaps for all module estimates in an \href{https://www.dropbox.com/s/891gvukhymr4a4b/30grps_50mods_est_0509.csv?dl=0}{online file}.

%\vspace{-1cm}
\section*{Acknowledgements}
We acknowledge support from NIH grant R01-GM130622 and helpful feedback from the Editors and two referees.

%\vspace{-1cm}

\newpage

%%\vspace{-1cm}
%\appendix 

%fill out appendices (old supplemental info) here
\begin{appendices}
\section{Additional methodological details}

\subsection{Notation details}
\label{app:notation}
Detailed explanations for our proposed model~\ref{model} in the main manuscript are listed in Table~\ref{tab:notation}. 
\begin{table}[H]
    \caption{Notation for the proposed model.}
    \centering
    \begin{tabular}{c|c|c}
    \hline
      $J$   & Observed &  Number of cohorts \\
      $K$  & Pre-specified & Number of covariate effects \\
     $L$   & Pre-specified & Number of auxiliary structures \\
    $\bX_j$  & Observed & Outcome matrix for  $j$th cohort \\
    $\bY_j$  & Observed & Covariate matrix for  $j$th cohort \\
    $\bY_{\cdot}^{(k)}$   & Pre-specified &  Design matrix for $k$th covariate effect concatenated from all $J$ cohorts \\
    $\bB_k$ & Estimated  & coefficients for $k$th covariate effect\\
    $\bS_{\cdot}^{(l)}$ & Estimated & $l$th auxiliary structure concatenated from all $J$ cohorts \\
    $\bE_j$ & Estimated & Random error matrix for the $j$th cohort \\
    $\bC_Y$  & Pre-specified &Binary indicator matrix where its $[j,k]$th entry\\ 
    & & determines whether $j$th cohort is considered in $k$th covariate effect\\
    $\bC_S$  & Pre-specified & Binary indicator matrix where its $[j,l]$th entry \\ 
    & & determines whether $j$th cohort is considered in $l$th auxiliary structure\\
    $\bU_S^{(l)}$ & Estimated & Loading matrix of $l$th auxiliary structure\\
    $\bV_{Sj}^{(l)}$ & Estimated & Score matrix of $l$th auxiliary structure for $j$th cohort\\
    $\bU_B^{(k)}$ & Estimated & Loading matrix of $k$th covariate coefficients \\
    $\bV_{B}^{(k)}$ & Estimated & Score matrix of $k$th covariate coefficients\\
         \hline
    \end{tabular}
    \label{tab:notation}
\end{table}

\subsection{Construction of module indicator matrices}
\label{app:indicator}

The construction of module indicator matrices $\bC_Y$ and $\bC_S$ can be accomplished in various ways depending on the specific context and available prior knowledge. If there exists prior knowledge indicating shared effects among certain cohorts, it would be straightforward to define modules accordingly. For example, defining a global module and individual modules for each cohort, as illustrated in Section~\ref{missingsims}. In absence of such information, there are other practical methods that can be applied.

In scenarios where the number of cohorts is small, one can begin by enumerating all possible combinations of cohorts in $\bC_Y$ and $\bC_S$.  As the objective function encourages rank sparsity, modules with no true shared structure may be estimated as zero (i.e., no structure) even if they are included  in the algorithm. Moreover, after obtaining initial estimates, modules that explain a relatively higher amount of variance can be retained to form updated $\bC_Y$ and $\bC_S$.

When dealing with a large number of cohorts, an alternative approach would be to adopt a data-driven strategy like the ``Optimization algorithm: dynamic modules" section from \cite{Lock2022BADIFACplus}. This method is designed to select $\bC_S$ based on the amount of variance explained, after which $\bC_Y$ can be set to match $\bC_S$, thus partitioning the variance related to covariates. This is illustrated in Section~\ref{sec:decomposition}. Optionally as a second step, one could keep the modules that explain a high amount of variance to reformulate $\bC_Y$ and $\bC_S$. This can help identify the most significant components in the covariate-related effects and auxiliary structures.

\section{Proofs}
\label{app:proof}

\subsection{Proof of Theorem 1}

\begin{proof}
Consider the following lemma, the proof of which is provided in \cite{mazumder2010spectral}: \begin{lem}\label{ZtoUV}
(Mazumder et al., 2010)
For any matrix $\bZ: m\times n$ with   $rank(\bZ) = k $, $\forall  r\geq k$, the following holds:
\begin{align*}
    || \bZ ||_{*} & = \min_{\substack{\bU,\bV:\\\bZ = \bU_{m\times r}\bV_{n\times r}^T}} \frac{1}{2} (|| \bU ||_{F}^2 + || \bV ||_{F}^2)
\end{align*}
\end{lem} 
Applying Lemma~\ref{ZtoUV} to each $\bB_k, k=1,...,K$ and $\bS_{\cdot}^{(l)}, l=1,...,L$: 
\begin{align*}
\min_{\{\bB_k\}^K_{k=1},\{\bS_{\cdot}^{(l)}\}^L_{l=1}} &\{ \frac{1}{2} ||\bX_{\cdot} - \sum^{K}_{k=1} \bB_k \bY_{\cdot}^{(k)} - \sum^{L}_{l=1}  \bS_{\cdot}^{(l)}||_F^2 + \sum^{K}_{k=1} \lambda_B^{(k)} ||\bB_k||_{*} +\sum^{L}_{l=1} \lambda_S^{(l)} ||\bS_{\cdot}^{(l)}||_{*}  \}\\
 =   \min_{\{\bB_k\}^K_{k=1},\{\bS_{\cdot}^{(l)}\}^L_{l=1}}& \{ \frac{1}{2}||\bX_{\cdot} - \sum^{K}_{k=1} \bB_k \bY_{\cdot}^{(k)} - \sum^{L}_{l=1}  \bS_{\cdot}^{(l)}||_F^2 + \\
  &  \sum^{K}_{k=1} \lambda_B^{(k)} \min_{\bU_B^{(k)},\bV_B^{(k)}:\bB_k=\bU_B^{(k)}\bV_B^{(k)T}}\frac{1}{2}(||\bU_B^{(k)}||_{F}^2 +||\bV_B^{(k)}||_{F}^2) +
    \\& \sum^{L}_{l=1} \lambda_S^{(l)}\min_{\bU_S^{(l)},\bV_S^{(l)}:\bS_{\cdot}^{(l)}=\bU_S^{(l)}\bV_S^{(l)T}}\frac{1}{2}( ||\bU_S^{(l)}||_{F}^2 + ||\bV_S^{(l)}||_{F}^2)  \}\\
 =   \min_{\substack{\{\bU_B^{(k)},\bV_B^{(k)}:\bB_k=\bU_B^{(k)}\bV_B^{(k)T}\}^K_{k=1},\\ \{\bU_S^{(l)},\bV_S^{(l)}:\bS_{\cdot}^{(l)}=\bU_S^{(l)}\bV_S^{(l)T}\}^L_{l=1}}} & \frac{1}{2}\{ ||\bX_{\cdot} - \sum^{K}_{k=1} \bU_B^{(k)}\bV_B^{(k)T} \bY_{\cdot}^{(k)} - \sum^{L}_{l=1}  \bU_S^{(l)}\bV_S^{(l)T}||_F^2 + \nonumber\\
    &\sum^{K}_{k=1} \lambda_B^{(k)} \min_{\bU_B^{(k)},\bV_B^{(k)}:\bB_k=\bU_B^{(k)}\bV_B^{(k)T}}(||\bU_B^{(k)}||_{F}^2 +||\bV_B^{(k)}||_{F}^2) +\\
    &\sum^{L}_{l=1} \lambda_S^{(l)}\min_{\bU_S^{(l)},\bV_S^{(l)}:\bS_{\cdot}^{(l)}=\bU_S^{(l)}\bV_S^{(l)T}}( ||\bU_S^{(l)}||_{F}^2 + ||\bV_S^{(l)}||_{F}^2)  \}\\
= \min_{\substack{\{\bU_B^{(k)},\bV_B^{(k)}\}^K_{k=1}, \{\bU_S^{(l)},\bV_S^{(l)}\}^L_{l=1}}}  &\frac{1}{2}\{ ||\bX_{\cdot} - \sum^{K}_{k=1} \bU_B^{(k)}\bV_B^{(k)T} \bY_{\cdot}^{(k)} - \sum^{L}_{l=1}  \bU_S^{(l)}\bV_S^{(l)T}||_F^2 + \nonumber\\
    &\sum^{K}_{k=1} \lambda_B^{(k)} (||\bU_B^{(k)}||_{F}^2 +||\bV_B^{(k)}||_{F}^2) +\sum^{L}_{l=1} \lambda_S^{(l)}( ||\bU_S^{(l)}||_{F}^2 + ||\bV_S^{(l)}||_{F}^2).  \}
\end{align*}

\end{proof}

\subsection{Proof of Proposition 1}

\begin{proof}
Assume a violation of condition 1, wherein $\lambda_B^{(k)} \geq \frac{1}{c_y} \sum_{i\in \mathcal{I}_k} \lambda_B^{(i)}$. Let $\widehat{\bB}_{k}\bY_{\cdot}^{(k)} = \sum_{i\in \mathcal{I}_k} \widehat{\bB}_{i}'\bY_{\cdot}^{(i)}$, where  $\widehat{\bB}_{i}'\bY_{\cdot}^{(i)}$ contains the blocks of $\widehat{\bB}_{k}\bY_{\cdot}^{(k)}$ corresponding to $\bC_Y[\cdot,i]$ and $\bold{0}$ otherwise. The choice of $\widehat{\bB}_{i}'$ is unique that $\widehat{\bB}_{i}' = \frac{1}{c_y}\widehat{\bB}_{k}$. For all $i\in \mathcal{I}_k$, we have $    ||\widehat{\bB}_{i}'\bY_{\cdot}^{(i)}||_{*}
    = ||\frac{1}{c_y}\widehat{\bB}_{k}\bY_{\cdot}^{(i)}||_{*} \leq||\frac{1}{c_y}\widehat{\bB}_{k}\bY_{\cdot}^{(k)}||_{*} \leq ||\widehat{\bB}_{k}\bY_{\cdot}^{(k)}||_{*}$
since $c_y\geq1$. 
Consider a minimizer $\{ \widetilde{\bB}_k \}_{k=1}^K, \{ \widetilde{\bS}_{\cdot}^{(l)} \}_{l=1}^L $, where $\{ \widetilde{\bS}_{\cdot}^{(l)} \}_{l=1}^L=\{ \widehat{\bS}_{\cdot}^{(l)} \}_{l=1}^L$, $\widetilde{\bB}_{k} = \bold{0}$, $\widetilde{\bB}_{i} = \widehat{\bB}_{i} + \widehat{\bB}_{i}', \forall i \in \mathcal{I}_k$,  and all other $\bB$ estimates are equal. Then, by the triangle inequality, 
\begin{align*}
       f(\{ \widehat{\bB}_k \}_{k=1}^K, \{ \widehat{\bS}_{\cdot}^{(l)} \}_{l=1}^L) - f(\{ \widetilde{\bB}_k \}_{k=1}^K, \{ \widetilde{\bS}_{\cdot}^{(l)} \}_{l=1}^L) 
     &  = \lambda_B^{(k)} ||\widehat{\bB}_k||_{*} +\sum_{i\in \mathcal{I}_k} \lambda_B^{(i)} ||\widehat{\bB}_i||_{*} - \sum_{i\in \mathcal{I}_k} \lambda_B^{(i)} ||\widehat{\bB}_i + \widehat{\bB}_i'||_{*}\\
    & \geq \lambda_B^{(k)} ||\widehat{\bB}_k||_{*} +\sum_{i\in \mathcal{I}_k} \lambda_B^{(i)} ||\widehat{\bB}_i||_{*} - \sum_{i\in \mathcal{I}_k} \lambda_B^{(i)} (||\widehat{\bB}_i||_{*} + ||\widehat{\bB}_i'||_{*})\\
    &  = \lambda_B^{(k)} ||\widehat{\bB}_k||_{*} - \sum_{i\in \mathcal{I}_k} \lambda_B^{(i)}  ||\widehat{\bB}_i'||_{*}\\
    & = (\lambda_B^{(k)} - \frac{1}{c_y}\sum_{i\in \mathcal{I}_k} \lambda_B^{(i)} ) ||\widehat{\bB}_k||_{*}\\
    & \geq 0
\end{align*}

Now assume a violation of condition 2, wherein $\lambda_B^{(k)} \geq \frac{1}{c_{sy}}\sum_{i\in \mathcal{I}_k} \lambda_S^{(i)} ||\bY_{\cdot}^{(k)}||_{*}$. Let $\widehat{\bB}_{k}\bY_{\cdot}^{(k)} = \sum_{i\in \mathcal{I}_k} \widehat{\bS}_{\cdot}^{(i)'}$, where  $\widehat{\bS}_{\cdot}^{(i)'}$ contains the blocks of $\widehat{\bB}_{k}\bY_{\cdot}^{(k)}$ corresponding to $\bC_Y[\cdot,i]$ and $\bold{0}$ otherwise. The choice of $\widehat{\bS}_{\cdot}^{(i)'}$ is unique that $\widehat{\bS}_{\cdot}^{(i)'}= \frac{1}{c_{sy}}\widehat{\bB}_{k}\bY_{S\cdot}^{(i)}$, where $\bY_{S\cdot}^{(i)} = [\bY_{S1}^{(i)},\bY_{S2}^{(i)},...,\bY_{S\bJ}^{(i)}]$ with  $\bY_{Sj}^{(i)} = $$\begin{dcases} 
      0_{q\times n_j} & \text{if }\bC_S[j,i] = 0 
      \\
  \bY_j & \text{if }  \bC_S[j,i] = 1
  \end{dcases}$ for all $i\in \mathcal{I}_k$. Since $\bY_{S\cdot}^{(i)}$ is gained by setting some blocks of $\bY_{\cdot}^{(k)}$ to be zero, $||\widehat{\bS}_{\cdot}^{(i)'}||_{*} = ||\frac{1}{c_{sy}}\widehat{\bB}_{k}\bY_{S\cdot}^{(i)}||_{*} \leq ||\frac{1}{c_{sy}}\widehat{\bB}_{k}\bY_{\cdot}^{(k)}||_{*}$. Consider a minimizer $\{ \widetilde{\bB}_k \}_{k=1}^K, \{ \widetilde{\bS}_{\cdot}^{(l)} \}_{l=1}^L $, where  $\widetilde{\bB}_{k} = \bold{0}$, $\widetilde{\bS}_{i} = \widehat{\bS}_{\cdot}^{(i)} + \widehat{\bS}_{\cdot}^{(i)'}, \forall i \in \mathcal{I}_k$,  and all other $\bB,\bS$ estimates are equal. Then, by the triangle inequality,
\begin{align*}
       f(\{ \widehat{\bB}_k \}_{k=1}^K, \{ \widehat{\bS}_{\cdot}^{(l)} \}_{l=1}^L) - f(\{ \widetilde{\bB}_k \}_{k=1}^K, \{ \widetilde{\bS}_{\cdot}^{(l)} \}_{l=1}^L) 
     &  = \lambda_B^{(k)} ||\widehat{\bB}_k||_{*} +\sum_{i\in \mathcal{I}_k} \lambda_S^{(i)} ||\widehat{\bS}_{\cdot}^{(i)}||_{*} - \sum_{i\in \mathcal{I}_k} \lambda_S^{(i)} ||\widehat{\bS}_{\cdot}^{(i)} + \widehat{\bS}_{\cdot}^{(i)'}||_{*}\\
    & \geq \lambda_B^{(k)} ||\widehat{\bB}_k||_{*} +\sum_{i\in \mathcal{I}_k} \lambda_S^{(i)} ||\widehat{\bS}_{\cdot}^{(i)}||_{*} - \sum_{i\in \mathcal{I}_k} \lambda_S^{(i)} (||\widehat{\bS}_{\cdot}^{(i)}||_{*} + ||\widehat{\bS}_{\cdot}^{(i)'}||_{*})\\
    &  = \lambda_B^{(k)} ||\widehat{\bB}_k||_{*} - \sum_{i\in \mathcal{I}_k} \lambda_S^{(i)}  ||\widehat{\bS}_{\cdot}^{(i)'}||_{*}\\
    &  \geq \lambda_B^{(k)} ||\widehat{\bB}_k||_{*} - \sum_{i\in \mathcal{I}_k} \lambda_S^{(i)}  ||\frac{1}{c_{sy}}\widehat{\bB}_{k}\bY_{\cdot}^{(k)}||_{*}\\
    &  \geq \lambda_B^{(k)} ||\widehat{\bB}_k||_{*} - \sum_{i\in \mathcal{I}_k} \lambda_S^{(i)}  ||\frac{1}{c_{sy}}\widehat{\bB}_{k}||_{*} ||\bY_{\cdot}^{(k)}||_{*}\\
    & = (\lambda_B^{(k)} - \frac{1}{c_{sy}}\sum_{i\in \mathcal{I}_k} \lambda_S^{(i)} ||\bY_{\cdot}^{(k)}||_{*}) ||\widehat{\bB}_k||_{*}\\
    & \geq 0
\end{align*}

The proof for condition 3 and 4 is similar to arguments made in \citep{Lock2022BADIFACplus}.
\end{proof}

\subsection{Proof of Proposition 2}

\begin{proof}
Since the rows of $\bY$ are linear independent, the projection onto the space spanned by rows of $\bY, \mathcal{R}(\bY),$ is $\bY^T(\bY\bY^T)^{-1}\bY = \bY^T\bY$. By decomposing $\bX$ onto $\mathcal{R}(\bY)$ and $\mathcal{R}(\bI-\bY)$, we have
\begin{align*}
  ||\bX-\bA\bY||_F^2 & = ||\bX(\bI- \bY^T\bY)+ \bX \bY^T\bY-\bA\bY||_F^2\\
  & = ||\bX(\bI- \bY^T\bY)||_F^2 + || \bX \bY^T\bY-\bA\bY||_F^2 +  tr[2 (\bI- \bY^T\bY) \bX^T (\bX \bY^T-\bA)\bY]\\
  & = ||\bX(\bI- \bY^T\bY)||_F^2 + || \bX \bY^T\bY-\bA\bY||_F^2 +  tr[2 \bY(\bI- \bY^T\bY) \bX^T (\bX \bY^T-\bA)]\\
  & = ||\bX(\bI- \bY^T\bY)||_F^2 + || \bX \bY^T\bY-\bA\bY||_F^2 
\end{align*}
The second equation comes from the cyclic property of trace. Construct an orthogonal matrix $\bQ = [\bY,\bY^*]$ where the rows of $\bY^*$ give an orthonormal basis for $\mathcal{R}(\bI-\bY)$, i.e. $\bY\bY^{*T}=\bold{0}$. Therefore,
\begin{align*}
|| \bX \bY^T\bY-\bA\bY||_F^2 & = tr[(\bX \bY^T\bY-\bA\bY)^T(\bX \bY^T\bY-\bA\bY)]\\
&=tr[(\bX \bY^T\bY-\bA\bY)^T(\bX \bY^T\bY-\bA\bY)\bQ^T\bQ]\\
&=|| \bX \bY^T\bY\bQ^T-\bA\bY\bQ^T||_F^2\\
&=|| \bX \bY^T\bY\begin{bmatrix}
\bY^T\\
\bY^{*T}
\end{bmatrix}-\bA\bY\begin{bmatrix}
\bY^T\\
\bY^{*T}
\end{bmatrix}||_F^2\\
&=|| \bX \bY^T\begin{bmatrix}
\bI\\
\bold{0}
\end{bmatrix}-\bA\begin{bmatrix}
\bI\\
\bold{0}
\end{bmatrix}||_F^2\\
&=||\bX\bY^T-\bA||_F^2
\end{align*}
Combining these results, we rewrite
\begin{align*}
 \min_{\bA} \{\frac{1}{2} ||\bX-\bA\bY||_F^2 + \lambda||\bA||_{*}\}
  &=\min_{\bA} \{\frac{1}{2}||\bX(\bI- \bY^T\bY)||_F^2+ \frac{1}{2} ||\bX\bY^T-\bA||_F^2 + \lambda||\bA||_{*}\}\\
    &= \min_{\bA} \{ \frac{1}{2} ||\bX\bY^T-\bA||_F^2 + \lambda||\bA||_{*}\} +\frac{1}{2}||\bX(\bI- \bY^T\bY)||_F^2
\end{align*}
Apply Lemma 1 and we get the first desired result. The second desired result follows immediately due to $||\bA||_* = ||\bA\bY||_*$ for $\bY\bY^T=\bI$. Let SVD of $\bA$ be $\bU_A\bD_A\bV_A^T$ and $\bV_A' = \bY^T\bV_A$. Since $\bV_A'^T\bV_A' = \bV_A^T\bY\bY^T\bV_A = \bV_A^T\bV_A = \bI$, $\bU_A\bD_A\bV_A'^T$ is the SVD of $\bA\bY$. $\bA$ and $\bA\bY$ share the same singular values so that their nuclear norms are the same.
\end{proof}

\subsection{Proof of Proposition 4}

\begin{proof}

We start with the special case when $\bY$ is an identity matrix with $q=n$ and $\bX=\bB+ \frac{1}{\sqrt{n}}\bE$. It follows directly from  \citep{Shabalin2013res} that \begin{align*}
\sigma_j(\bX)\xrightarrow{P}
    \begin{dcases} \sqrt{1+\sigma_j^2(\bB)+c+\frac{c}{\sigma_j^2(\bB)}},
   & \text{if }  \sigma_j(\bB)>\sqrt[4]{c} \\ 
  1+\sqrt{c}, & \text{if }  \sigma_j(\bB)\leq\sqrt[4]{c} 
  \end{dcases}
\end{align*}. Note $s(\sigma_j(\bB)) = \sqrt{1+\sigma_j^2(\bB)+c+\frac{c}{\sigma_j^2(\bB)}}$ is a monotonic increasing function when $\sigma_j(\bB)>\sqrt[4]{c}$. Therefore, $min \{s(\sigma_j(\bB))\} > s(\sqrt[4]{c}) = 1 + \sqrt{c}$.

Now consider the more general case that $\bX_{m\times n}=\bB_{m\times q}\bY_{q\times n}+\bE_{m\times n}$. Due to semi-orthogonality of $\bY$, $\bX\bY^T=\bB+\bE\bY^T$. Since $\bE$ is of matrix normal distribution  $\mathcal{MN}_{m,n}(\bold{0}_{m\times n},\bI_{m\times m},\bI_{n\times n})$ and $\bY^T_{n\times q}$ is a linear transformation of full rank $q\leq n$, we have 
\begin{align*}
    \bE \bY^T \sim \mathcal{MN}_{m,q}(\bold{0}\bY^T,\bI_{m\times m},\bY \bI \bY^T) = \mathcal{MN}_{m,q}(\bold{0}_{m\times q},\bI_{m\times m},\bI_{q\times q}).
\end{align*}
Recognize that entries of $\bE\bY^T$ are still independent normal. Denote $\Bar{\bE} = \bE\bY^T$ and $\Bar{\bX} = \bX\bY^T$. The original question becomes $\Bar{\bX}_{m\times q} = \bB_{m\times q} + \Bar{\bE}_{m\times q}$. Applying the result of the special case, 
\begin{align*}
\sigma_j(\bX\bY^T) = 
\sigma_j(\Bar{\bX})= \xrightarrow{P}
    \begin{dcases} \sqrt{1+\sigma_j^2(\bB)+c+\frac{c}{\sigma_j^2(\bB)}},
   & \text{if }  \sigma_j(\bB)>\sqrt[4]{c} \\ 
  1+\sqrt{c}, & \text{if }\sigma_j(\bB)\leq\sqrt[4]{c}.  \hspace{80 pt}  %\qedsymbol{} 
  \end{dcases} 
\end{align*}  
Following the aforementioned reasoning, we have the conclusion that $\sigma_j(\bX\bY^T)$ will converge to a number larger than $1+\sqrt{c}$ as $\sigma_j(\bB)>\sqrt[4]{c}$.
\end{proof}

\subsection{Proposition 5 and its proof}

\begin{prop}\label{sameBYS}
For any semi-orthogonal matrix $\bY$ such that $\bY^T\bY = \bI$, if the optimization problems $\min_{\bB} \{\frac{1}{2} ||\bX-\bB\bY||_F^2 + \lambda||\bB||_{*}\}$ and $\min_{\bS} \{\frac{1}{2} ||\bX-\bS||_F^2 + \lambda||\bS||_{*}\}$ have their optimal solutions as $\widehat{\bB}$ and $\widehat{\bS}$ respectively, then $\widehat{\bS} = \widehat{\bB}\bY$.
\end{prop}

\begin{proof}
 Since orthogonal transformation preserves Frobenius norm, we have $||\bX-\bB\bY||_F = ||\bX\bY^T-\bB||_F$. Then, 
\begin{align*}
    \min_{\bB} \{\frac{1}{2} ||\bX-\bB\bY||_F^2 + \lambda||\bB||_{*}\} = \min_{\bB} \{\frac{1}{2} ||\bX\bY^T-\bB||_F^2 + \lambda||\bB||_{*}\}.
\end{align*}
Let SVD of $\bX$ to be $\bU\bD\bV^T$. Since $(\bY\bV)^T(\bY\bV) = \bV^T\bY^T\bY\bV = \bV^T\bV = \bI$, we have $\bU\bD(\bY\bV)^T = \bU\bD\bV^T\bY^T = \bX\bY^T$ serving as the SVD of $\bX\bY^T$.
Applying Lemma 1, the solution for $\min_{\bB} \{\frac{1}{2} ||\bX\bY^T-\bB||_F^2 + \lambda||\bB||_{*}\}$ is $\widehat{\bB} = \bU\widetilde{\bD}(\bY\bV)^T$ and the solution for $\min_{\bS} \{\frac{1}{2} ||\bX-\bS||_F^2 + \lambda||\bS||_{*}\}$ is $\widehat{\bS} = \bU\widetilde{\bD}\bV^T$. Therefore, $\widehat{\bS} = \widehat{\bB}\bY$.\\
\end{proof}

If $\bY_{\cdot}^{(k)}$ is semi-orthogonal with $\bY_{\cdot}^{(k)T}\bY_{\cdot}^{(k)} = \bI_n (q\geq n)$ and contains information from exactly the same cohorts as $\bS_{\cdot}^{(l)}$ (i.e., $\bC_Y[\cdot,k] = \bC_S[\cdot,l]$), the estimation process cannot differentiate $\widehat{\bB}_k\bY_{\cdot}^{(k)}$ from $\widehat{\bS}_{\cdot}^{(l)}$. Even without overlapping modules, the newly proposed model reverts to the unsupervised model~\ref{ObjUV} described in Section~\ref{sec:obejective}, which only comprises $\{ \widehat{\bS}{\cdot}^{(l)} \}_{l=1}^L$. Thus, the loss with semi-orthogonal regressors $\bY: q\times n$ is only valid if $q<n$. This finding aligns with our real-world problem of interest: in reality, we have only a few somatic mutations, the number of which is less than the number of patients in the study.

\section{Scaling and orthogonalization}
\label{app:scaling}

In reality the original $\bY$ is not orthogonal. In practice we scale $\bX$ and orthogonalize $\bY$ prior to optimization. We first center each row of $\bX_{\cdot}$ to have mean 0. In order to satisfy the standard normal noise requirement, we estimate the error variance for $\bX_{\cdot}$ by using the median absolute deviation estimator from \citep{gavish2017optimal}. The estimated variance of $\bX_{\cdot}$ is denoted as $\hat{\sigma}^2$. Then, we use $\bX_{\cdot}/\hat{\sigma}$ as the final data matrix for optimization, which has residual variance approximately 1. 
We further orthogonalize the columns of $\bY$ via SVD prior to optimization, and transform the solution back to the original covariate space afterward.
 %For original $\bY$ is not orthogonal, there will be a slight difference in standardizing or orthogonalizing the original $\bY$. 
 Here, we describe two scaling approaches for $\bY$: one standardizing the original covariates with out orthogonalization, and another rotating the covariate space so that it is orthogonal. In the next section, we show simulation results that illustrate the difference between these two approaches. 

Instead of estimating $\bB$ based on original $\bY_{\cdot}^{(k)}$, we first center each covariate in $\bY_{\cdot}^{(k)}$ and scale each covariate to have variance 1 and then do the optimization. For each $k=1,...,K$, the detailed steps are:
\begin{enumerate}
    \item Center each covariate in $\bY_{\cdot}^{(k)}$ and calculate the square root of row sums of squared entry of $\bY_{\cdot}^{(k)}$, denoted as $t_1^{(k)},t_2^{(k)},...,t_q^{(k)}$.
    \item Construct a diagonal matrix $\bT_k = diag(t_1^{(k)},...,t_q^{(k)})$. 
    \item Use $\bY_{\cdot scaled}^{(k)} = \bT_k^{-1}\bY_{\cdot}^{(k)}$ in optimization.
    \item Denote the estimation with respect to $\bY_{\cdot scaled}^{(k)}$ as $\widehat{\bB}_{k,scaled}$. Then, our final estimation for $\bB_k$ on its original scale is $\widehat{\bB}_k = \widehat{\bB}_{k,scaled} \bT_k$.
\end{enumerate}
 In order to further reduce collinearity among covariates, after centerization we may instead choose to apply orthogonalization to $\bY_{\cdot}^{(k)}$ rather than standardization. For each $k=1,...,K$, the detailed steps are:
 \begin{enumerate}
\item Obtain the SVD of $\bY_{\cdot}^{(k)}$, i.e. $\bY_{\cdot}^{(k)} = \bU_k \bD_k \bV_k^T$.
\item Use $\bY_{\cdot orth}^{(k)} = \bV_k^T$ in optimization.
\item Denote the estimation with respect to $\bY_{\cdot orth}^{(k)}$ as $\widehat{\bB}_{k,orth}$. Then, our final estimation for $\bB_k$ on its original scale is $\widehat{\bB}_k = \widehat{\bB}_{k,orth} \bU_k \bD_k$.
 \end{enumerate}
There are other ways to orthogonalize, such as Gram-Schmidt method. No matter standardization or orthogonalization, we record the transforming matrix from the original $\bY_{\cdot}^{(k)}$ to new $\bY_{\cdot scaled/orth}^{(k)}$ for correction of estimated $\bB_k$ at the final stage.

\section{More details on the simulations}

\subsection{Complete data generation}
\label{app:data_gen}
Here we describe the complete process to generate data in Section~\ref{missingsims}.

\begin{enumerate}
    \item %If the correlations among each feature in $\bY$ is zero, 
    Every entry in each $\bY_j, j= 1,...,J=30$ is drawn independently from a standard normal distribution. By default, the variance of each feature is 1. Construct one global shared $\bY_{\cdot}^{(1)} = [\bY_1,...,\bY_{30}]$ and 30 individual $\bY_{\cdot}^{(k)} = [\bold{0},...,\bY_{k-1},...,\bold{0}], k=2,...,31$ (only the $k-1$th submatrix is non-zero). 
    \item For the $K=31$ modules of $\bY^{(k)}_{\cdot}$, generate $\bB_{k} = \bU_{B}^{(k)} \bV_{B}^{(k)T} /sd(\bU_{B}^{(k)} \bV_{B}^{(k)T} \bY_{\cdot}^{(k)}) *\sqrt{n_k/n}$, where $n_k$ is the number of samples in module $k$. Each entry of $\bU_{B}^{(k)}: p\times r, \bV_{B}^{(k)}: q \times r$ comes from standard Normal distribution.
    \item  For a number of $L =31$ modules of $\bS^{(l)}_{\cdot}$ involved, draw one global score matrix $\bV_{S}^{(1)}: n\times r$ and 30 individual score matrices $\bV_{S}^{(l)} = [\bold{0},...,\bV_{l-1},...,\bold{0}], l=2,...,31$, where each entry of $\bV_1,...,\bV_{30}$ comes from standard Normal. Generate $\bS_{\cdot}^{(l)} = \bU_{S}^{(l)} \bV_{S}^{(l)T} /sd(\bU_{S}^{(l)} \bV_{S}^{(l)T} ) *\sqrt{n_l/n}$, where each entry of $\bU_{S}^{(l)}: p\times r$ comes from standard Normal distribution and $n_l$ is the number of samples in module $l$.
    \item Draw each entry of $\bE_{\cdot}$ from a standard normal distribution.
    \item Generate $\bX_{\cdot} = a * \bB_1 \bY_{\cdot}^{(1)} + b* \bS_{\cdot}^{(1)} + c* \sum^{31}_{k=2} \bB_k \bY_{\cdot}^{(k)}  + d*\sum^{31}_{l=2}  \bS_{\cdot}^{(l)} + \bE_{\cdot} $. The letters $a,b,c,d$ are constant for signal size. For instance,  scenario ($a$), $ a = \sqrt{10}$ and the remaining equal 1. 
\end{enumerate}

\subsection{Computation time}
\label{app:computation}

Table~\ref{tab:computation_time} demonstrates the computation time for all methods in missing data imputation for the TCGA data. The computation time of non-missing optimization and missing data imputation are similar. The proposed maRRR method consumes the most time since it considers both covariate effects and auxiliary modules. Generally, the computation time is proportional to the number of modules. But it will vary because of different number of cohorts within one module. Algorithm 1 requests more computation time when the true rank in some module is large. In general, both algorithms output similar RSE so that the one takes less computation time is used in the simulation. The case of missing columns takes slightly less computation time than the other two cases. This results from that it is uninformative when missing certain subjects. It may lead to all-zero imputation and this explains why aRRR for 30 separate modules takes significantly less time than the other cases. 

\begin{table}[H]
\caption{Computation time (in seconds) for our proposed methods and all other comparison methods in missing data imputation (based on algorithm 1 and 2 respectively) for the TCGA real data application. Each other method is based on 30 epochs. NA represents that missing data imputation based on Algorithm 2 does not work for ``NNreg" and ``NNapprox". ``50 modules" means that the method is based on 50 detected modules. ``1 all-shared + 30 separate" means that the method is based on one all-shared module and 30 separate modules (in total 31 modules). For method name representations, please refer to Section~\ref{missingsims} in the main context.}
\begin{tabular}{l|llll|llll}
\hline
                                                                                  & \multicolumn{4}{l|}{Algorithm 1}                                                                                                                                                         & \multicolumn{4}{l}{Algorithm 2}                                                                                                                                                          \\ \hline
Method                                                                            & \begin{tabular}[c]{@{}l@{}}missing\\ entries\end{tabular} & \begin{tabular}[c]{@{}l@{}}missing\\ columns\end{tabular} & \begin{tabular}[c]{@{}l@{}}missing\\ rows\end{tabular} & average & \begin{tabular}[c]{@{}l@{}}missing\\ entries\end{tabular} & \begin{tabular}[c]{@{}l@{}}missing\\ columns\end{tabular} & \begin{tabular}[c]{@{}l@{}}missing\\ rows\end{tabular} & average \\ \hline
\begin{tabular}[c]{@{}l@{}}maRRR\\ 50 modules\end{tabular}                        & 6251.0                                                    & 5345.2                                                    & 6442.4                                                 & 6012.9  & 3725.3                                                    & 3264.7                                                    & 3754.8                                                 & 3581.6  \\ \cline{1-1}
\begin{tabular}[c]{@{}l@{}}BIDIFAC+,\\ 50 modules\end{tabular}                    & 5000.7                                                    & 4283.4                                                    & 5198.2                                                 & 4827.4  & 1140.7                                                    & 1058.2                                                    & 1161.3                                                 & 1120.1  \\ \cline{1-1}
\begin{tabular}[c]{@{}l@{}}mRRR,\\ 50 modules\end{tabular}                        & 1566.7                                                    & 1327.9                                                    & 1576.9                                                 & 1490.5  & 2905.9                                                    & 2519.9                                                    & 2873.3                                                 & 2766.4  \\ \cline{1-1}
\begin{tabular}[c]{@{}l@{}}maRRR,\\ 1 all-shared + \\ 30 separate\end{tabular}    & 1925.6                                                    & 1623.3                                                    & 1962.9                                                 & 1837.3  & 2276.5                                                    & 2085.1                                                    & 2303.6                                                 & 2221.7  \\ \cline{1-1}
\begin{tabular}[c]{@{}l@{}}BIDIFAC+,\\ 1 all-shared + \\ 30 separate\end{tabular} & 1109.9                                                    & 967.1                                                     & 1124.9                                                 & 1067.3  & 705.2                                                     & 684.8                                                     & 714.9                                                  & 701.6   \\ \cline{1-1}
\begin{tabular}[c]{@{}l@{}}mRRR,\\ 1 all-shared + \\ 30 separate\end{tabular}     & 938.2                                                     & 823.7                                                     & 939.5                                                  & 900.4   & 1711.2                                                    & 1556.3                                                    & 1729.0                                                 & 1665.5  \\ \cline{1-1}
\begin{tabular}[c]{@{}l@{}}aRRR, one\\ all-shared\end{tabular}                    & 617.9                                                     & 524.8                                                     & 637.8                                                  & 593.5   & 75.1                                                      & 67.6                                                      & 74.7                                                   & 72.5    \\ \cline{1-1}
\begin{tabular}[c]{@{}l@{}}aRRR, \\ 30 separate\end{tabular}                      & 1246.3                                                    & 68.5                                                      & 1249.0                                                 & 854.6   & 2196.8                                                    & 123.4                                                     & 2195.6                                                 & 1505.3  \\ \cline{1-1}
\begin{tabular}[c]{@{}l@{}}NNreg, one\\ all-shared\end{tabular}                   & 2.2                                                       & 2.2                                                       & 5.5                                                    & 3.3     & NA                                                        & NA                                                        & NA                                                     & NA      \\ \cline{1-1}
\begin{tabular}[c]{@{}l@{}}NNapprox,\\ one all-shared\end{tabular}                & 116.7                                                     & 49.6                                                      & 585.0                                                  & 250.5   & NA                                                        & NA                                                        & NA                                                     & NA      \\ \cline{1-1}
\begin{tabular}[c]{@{}l@{}}NNreg,\\ 30 separate\end{tabular}                      & 4.5                                                       & 18.5                                                      & 21.5                                                   & 14.9    & NA                                                        & NA                                                        & NA                                                     & NA      \\ \cline{1-1}
\begin{tabular}[c]{@{}l@{}}NNapprox, \\ 30 separate\end{tabular}                  & 60.4                                                      & 22.6                                                      & 256.2                                                  & 113.1   & NA                                                        & NA                                                        & NA                                                     & NA      \\ \hline
\end{tabular}
\label{tab:computation_time}
\end{table}

\subsection{Additional simulation to assess aRRR and maRRR}
\label{app:additional_sims}

Here we discuss an additional simulation study to assess aspects of aRRR and maRRR, including the effects of orthogonalizing Y and the recovery of the underlying ranks of the true structure. The simulation covers scenarios in which the original explanatory data matrices $\bY_j, j= 1,...,J$ are orthogonal or not. The process, used to generate the complete data %for section 6.2 in the manuscript, 
for this section, is as follows:
\begin{enumerate}
    \item %If the correlations among each feature in $\bY$ is zero, 
    Every entry in each $\bY_j, j= 1,...,J$ is drawn independently from a standard normal distribution. 
    %If their correlation matrix is given, draw each subject's observation independently from a normal distribution with mean $\boldmath{0}$ and the corresponding correlation matrix. 
    By default, the variance of each feature is 1. Denote the standard deviation for each $\bY_j, j= 1,...,J$ as $sd(\bY_j), j= 1,...,J$.
    \item If we aim to orthogonalize some $\bY_j, j= 1,...,J$, denote the transpose of the right orthogonal matrix from singular value decomposition of $\bY_j$ as $\bV_{Y_j}$. Calculate the orthogonal version of $\bY_j, j= 1,...,J$ by $\bV_{Y_j}/sd(\bY_j), j= 1,...,J$ in order to keep the original scale of variation.
    \item For a number of $K$ modules of covariate effects involved, construct $\bY^{(k)}_{\cdot}, k=1,...,K$. Generate $\bB_{k} = \bU_{B}^{(k)} \bV_{B}^{(k)T}$, where each entry of $\bU_{B}^{(k)}: q\times r, \bV_{B}^{(k)}: n_k \times r$ comes from standard Normal distribution.
    \item For a number of $L$ modules of $\bS^{(l)}_{\cdot}$ involved, generate $\bS_{\cdot}^{(l)} = \bU_{S}^{(l)} \bV_{S}^{(l)T}$, where each entry of $\bU_{S}^{(l)}: p\times r$ and each non-zero entry of $\bV_{S}^{(l)}: q\times r$ comes from standard Normal distribution.
    \item Draw each entry of $\bE_{\cdot}$ from a standard normal distribution.
    \item Generate $\bX_{\cdot} = \sum^{K}_{k=1} \bB_k \bY_{\cdot}^{(k)} + \sum^{L}_{l=1}  \bS_{\cdot}^{(l)} + \bE_{\cdot} $.
\end{enumerate}

Besides mean squared error (mse) as a metric to assess accuracy, we want to understand how low-rank structures for $\bB, \bS$ are uncovered is estimated under different orthogonality.  Therefore, we define the "rank sum ratio" as follows:
\begin{align*}
   \frac{\sum_{i=1}^{r_B} \lambda_i (\widehat{\bB})  }{\sum_{i=r_B + 1}^{r_{B,upper}} \lambda_i (\widehat{\bB}) }, \frac{\sum_{i=1}^{r_S} \lambda_i (\widehat{\bS})  }{\sum_{i=r_S + 1}^{r_{S,upper}} \lambda_i (\widehat{\bS}) }
\end{align*}
where $r_B$ is the true rank of $\bB$ and $r_{B,upper}$ is the upper bound of $\bB$ specified in estimation, both similarly defined for $\bS$. The smaller the rank sum ratio is, the lower the rank structure achieves. 

Table~\ref{tab:arrr_orth} and \ref{tab:marrr_orth} shows the MSE and rank sum ratio for aRRR and maRRR simulations respectively. We consider one cohort for aRRR and two cohorts for maRRR. In general, orthogonaliztion before optimization and in data generation achieve similar results. Compared with only standardization of $\bY$, orthogonaliztion of $\bY$ leads to less MSE and more accurate rank estimations. After orthogonaliztion, the proposed methods still overestimate the true rank of covariate-related signal, but it is not severe as the rank sum ratio is below 0.01. Therefore, orthogoalization of $\bY$ is helpful for optimization. 

\begin{table}[H]
\caption{Mean squared error(MSE) for aRRR under different scenarios of orthogonality, true rank and signal size. Row names: ``r\_y" refers to true rank of $\bY$ in the data generation process. ``sd\_YB" and ``sd\_S" refer to the standard deviation for matrix $\bB \bY$ and $\bS$ in generation respectively. ``epochs" refers to the number of iterations to converge. ``ratio\_B" refers to the rank sum ratio defined before. ``est\_rank\_B" counts the number of singular values that large than 0.1. Column names: ``no orth" means that the original $\bY$ is not orthogonal and we only standardize it before optimization. ``orth\_opt" means that we only orthogonalize original ``$\bY$" before optimization. ``orth\_gen" means that the original $\bY$ is orthogonal. }
\begin{tabular}{lllllllll}
\hline
               & r\_y & sd\_YB & sd\_S & epochs & mse\_B & mse\_S & est\_rank\_B & ratio\_B \\ \hline
no orth        & 1    & 5      & 0.5   & 65.82  & 0.001  & 0.607  & 1.3          & 0.002    \\
no orth        & 1    & 1      & 1     & 31.27  & 0.038  & 0.225  & 1.39         & 0.012    \\
no orth        & 1    & 0.5    & 5     & 44.82  & 0.172  & 0.012  & 1.48         & 0.037    \\
no orth        & 5    & 5      & 0.5   & 52.36  & 0.007  & 0.629  & 5.04         & 0        \\
no orth        & 5    & 1      & 1     & 34.38  & 0.141  & 0.238  & 4.95         & 0.001    \\
no orth        & 5    & 0.5    & 5     & 41.5   & 0.404  & 0.013  & 4.42         & 0        \\ \hline
orth\_opt      & 1    & 5      & 0.5   & 63.66  & 0.001  & 0.606  & 1.06         & 0        \\
orth\_opt      & 1    & 1      & 1     & 33.69  & 0.033  & 0.224  & 1.12         & 0.004    \\
orth\_opt      & 1    & 0.5    & 5     & 47.4   & 0.157  & 0.012  & 1.16         & 0.011    \\
orth\_opt      & 5    & 5      & 0.5   & 45.41  & 0.006  & 0.626  & 5            & 0        \\
orth\_opt      & 5    & 1      & 1     & 27.62  & 0.125  & 0.238  & 4.94         & 0        \\
orth\_opt      & 5    & 0.5    & 5     & 39.7   & 0.372  & 0.013  & 4.42         & 0        \\ \hline
orth\_gen      & 1    & 5      & 0.5   & 63.31  & 0.001  & 0.606  & 1.08         & 0.001    \\
orth\_gen      & 1    & 1      & 1     & 33.47  & 0.034  & 0.224  & 1.1          & 0.004    \\
orth\_gen      & 1    & 0.5    & 5     & 47.82  & 0.161  & 0.012  & 1.14         & 0.012    \\
orth\_gen      & 5    & 5      & 0.5   & 45.42  & 0.006  & 0.624  & 5            & 0        \\
orth\_gen      & 5    & 1      & 1     & 27.72  & 0.122  & 0.237  & 4.96         & 0        \\
orth\_gen      & 5    & 0.5    & 5     & 41.2   & 0.369  & 0.013  & 4.34         & 0.001    \\ \hline
\end{tabular}
\label{tab:arrr_orth}
\end{table}

\begin{table}[H]
\caption{Mean squared error(MSE) for maRRR under different scenarios of orthogonality, true rank and signal size. Row names: ``r\_y" refers to true rank of $\bY$ in the data generation process. ``sd\_YB" and ``sd\_S" refer to the standard deviation for matrix $\bB \bY$ and $\bS$ in generation respectively. ``epochs" refers to the number of iterations to converge.  ``ratio\_B" refers to the rank sum ratio defined before. ``est\_rank\_B" counts the number of singular values that large than 0.1. Column names: ``no orth" means that the original $\bY$ is not orthogonal and we only standardize it before optimization. ``orth\_opt" means that we only orthogonalize original ``$\bY$" before optimization. ``orth\_gen" means that the original $\bY$ is orthogonal. All results are the average of all $\bB$ or $\bS$.}
\begin{tabular}{lllllllll}
\hline
               & r\_y & sd\_YB & sd\_S & epochs & mse\_B & mse\_S & est\_rank\_B & ratio\_B \\ \hline
no orth        & 1    & 2      & 0.2   & 84.56  & 0.008  & 0.953  & 1.493        & 0.008    \\
no orth        & 1    & 1      & 1     & 77.72  & 0.037  & 0.175  & 1.873        & 0.024    \\
no orth        & 1    & 0.2    & 2     & 88.6   & 0.4    & 0.065  & 1.76         & 0.141    \\
no orth        & 5    & 2      & 0.2   & 123.45 & 0.115  & 0.962  & 6.617        & 0.033    \\
no orth        & 5    & 1      & 1     & 87.78  & 0.202  & 0.195  & 6.06         & 0.023    \\
no orth        & 5    & 0.2    & 2     & 91.41  & 0.735  & 0.065  & 3.273        & 0        \\ \hline
orth\_opt      & 1    & 2      & 0.2   & 82.34  & 0.007  & 0.952  & 1.393        & 0.007    \\
orth\_opt      & 1    & 1      & 1     & 76.96  & 0.034  & 0.174  & 1.66         & 0.019    \\
orth\_opt      & 1    & 0.2    & 2     & 90.04  & 0.366  & 0.065  & 1.543        & 0.089    \\
orth\_opt      & 5    & 2      & 0.2   & 119.61 & 0.114  & 0.955  & 6.727        & 0.036    \\
orth\_opt      & 5    & 1      & 1     & 86.66  & 0.194  & 0.195  & 6.18         & 0.026    \\
orth\_opt      & 5    & 0.2    & 2     & 93.63  & 0.702  & 0.065  & 3.313        & 0        \\ \hline
orth\_gen      & 1    & 2      & 0.2   & 82.64  & 0.007  & 0.952  & 1.333        & 0.006    \\
orth\_gen      & 1    & 1      & 1     & 78.4   & 0.034  & 0.175  & 1.64         & 0.018    \\
orth\_gen      & 1    & 0.2    & 2     & 89.62  & 0.367  & 0.065  & 1.52         & 0.09     \\
orth\_gen      & 5    & 2      & 0.2   & 116.33 & 0.109  & 0.955  & 6.743        & 0.035    \\
orth\_gen      & 5    & 1      & 1     & 86.79  & 0.186  & 0.195  & 6.2          & 0.027    \\
orth\_gen      & 5    & 0.2    & 2     & 93.07  & 0.692  & 0.065  & 3.437        & 0        \\ \hline
\end{tabular}
\label{tab:marrr_orth}
\end{table}

\section{More details on the real data analysis}

\subsection{Cancer type and mutation details}
\label{app:names}

We provide Table~\ref{tab:cancer_types} as the summary of all the 30 cancer types and Table~\ref{tab: mutation} as the summary of all the 50 somatic mutations considered in our real data analysis (Section~\ref{sec:data}).
\begin{table}[H]
\caption{The cancer study abbreviations, sample sizes and study names, sourcing from National Cancer Institute \url{https://gdc.cancer.gov/resources-tcga-users/tcga-code-tables/tcga-study-abbreviations}.}
\begin{tabular}{lll}
\hline
Abbreviation & Sample Size & Study Name                                                       \\\hline
ACC          & 77          & Adrenocortical carcinoma                                         \\
BLCA         & 129         & Bladder Urothelial Carcinoma                                     \\
BRCA         & 976         & Breast invasive carcinoma                                        \\
CESC         & 193         & Cervical squamous cell carcinoma and endocervical adenocarcinoma \\
COAD         & 147         & Colon adenocarcinoma                                             \\
ESCA         & 184         & Esophageal carcinoma                                             \\
GBM          & 150         & Glioblastoma multiforme                                          \\
HNSC         & 279         & Head and Neck squamous cell carcinoma                            \\
KICH         & 66          & Kidney Chromophobe                                               \\
KIRC         & 415         & Kidney renal clear cell carcinoma                                \\
KIRP         & 161         & Kidney renal papillary cell carcinoma                            \\
LAML         & 170         & Acute Myeloid Leukemia                                           \\
LGG          & 283         & Brain Lower Grade Glioma                                         \\
LIHC         & 195         & Liver hepatocellular carcinoma                                   \\
LUAD         & 230         & Lung adenocarcinoma                                              \\
LUSC         & 178         & Lung squamous cell carcinoma                                     \\
OV           & 115         & Ovarian serous cystadenocarcinoma                                \\
PAAD         & 150         & Pancreatic adenocarcinoma                                        \\
PCPG         & 179         & Pheochromocytoma and Paraganglioma                               \\
PRAD         & 331         & Prostate adenocarcinoma                                          \\
READ         & 64          & Rectum adenocarcinoma                                            \\
SARC         & 245         & Sarcoma                                                          \\
SKCM         & 342         & Skin Cutaneous Melanoma                                          \\
STAD         & 275         & Stomach adenocarcinoma                                           \\
TGCT         & 149         & Testicular Germ Cell Tumors                                      \\
THCA         & 400         & Thyroid carcinoma                                                \\
THYM         & 119         & Thymoma                                                          \\
UCEC         & 242         & Uterine Corpus Endometrial Carcinoma                             \\
UCS          & 57          & Uterine Carcinosarcoma                                           \\
UVM          & 80          & Uveal Melanoma\\    \hline                                  
\end{tabular}
\label{tab:cancer_types}
\end{table}

\begin{table}[H]
\caption{Gene labels for 50 somatic mutation data in order.}
\centering
\begin{tabular}{cccccc}
\hline
Index & Mutation & Index & Mutation & Index & Mutation \\\hline
1     & TP53     & 18    & FAT4     & 35    & PKHD1L1  \\
2     & TTN      & 19    & HMCN1    & 36    & RYR1     \\
3     & MUC16    & 20    & CSMD1    & 37    & RYR3     \\
4     & PIK3CA   & 21    & MUC5B    & 38    & NEB      \\
5     & CSMD3    & 22    & ZFHX4    & 39    & PCDH15   \\
6     & LRP1B    & 23    & FAT3     & 40    & DST      \\
7     & KRAS     & 24    & SPTA1    & 41    & MLL3     \\
8     & RYR2     & 25    & GPR98    & 42    & MLL2     \\
9     & MUC4     & 26    & PTEN     & 43    & MACF1    \\
10    & FLG      & 27    & FRG1B    & 44    & DNAH9    \\
11    & SYNE1    & 28    & AHNAK2   & 45    & BRAF     \\
12    & USH2A    & 29    & APOB     & 46    & DNAH11   \\
13    & PCLO     & 30    & ARID1A   & 47    & DNAH8    \\
14    & APC      & 31    & LRP2     & 48    & CSMD2    \\
15    & DNAH5    & 32    & XIRP2    & 49    & MUC2     \\
16    & OBSCN    & 33    & Unknown  & 50    & ABCA13   \\
17    & MUC17    & 34    & DMD      &       &         \\\hline
\end{tabular}
\label{tab: mutation}
\end{table}

\subsection{Data processing and distributions}
\label{app:data_process}

The pan-cancer RNASeq data, described in \cite{hoadley2018cell}, were downloaded as the file `EBPlusPlusAdjustPANCAN\_IlluminaHiSeq\_RNASeqV2.geneExp.tsv' from \url{https://gdc.cancer.gov/about-data/publications/PanCan-CellOfOrigin} [accessed June 23, 2021].  This dataset had undergone preprocessing steps described in \cite{hoadley2018cell}, including batch correction using an empirical Bayes approach and upper-quartile normalization.  These data were further log-transformed (via a log$(1+x)$ transformation), and filtered to the 1000 genes with the highest standard deviation after log-transformation.   The log-transformed and filtered data were then gene-centered by subtracting the overall mean (across all cancer types) for each gene.  The distribution of processed expression values for each cancer type are shown in Figure \ref{fig:rna_violin};  the distributions are roughly similar and approximately bell-shaped across the different cancer types.  

The somatic mutation data were binary, prior to the the scaling described in Section \ref{app:scaling}, where `1' implies there is a somatic mutation in the gene for the given sample and `0' implies there is no somatic mutation.  Figure \ref{fig:mut_violin} gives the proportion of samples that have a mutation across the 50 genes considered, for each cancer type.  Some cancer types have several genes that are frequently mutated, while others have a sparser profile with no frequently mutated genes among those considered.      

\begin{figure}[H] 
        \centering
        \includegraphics[width=\textwidth]{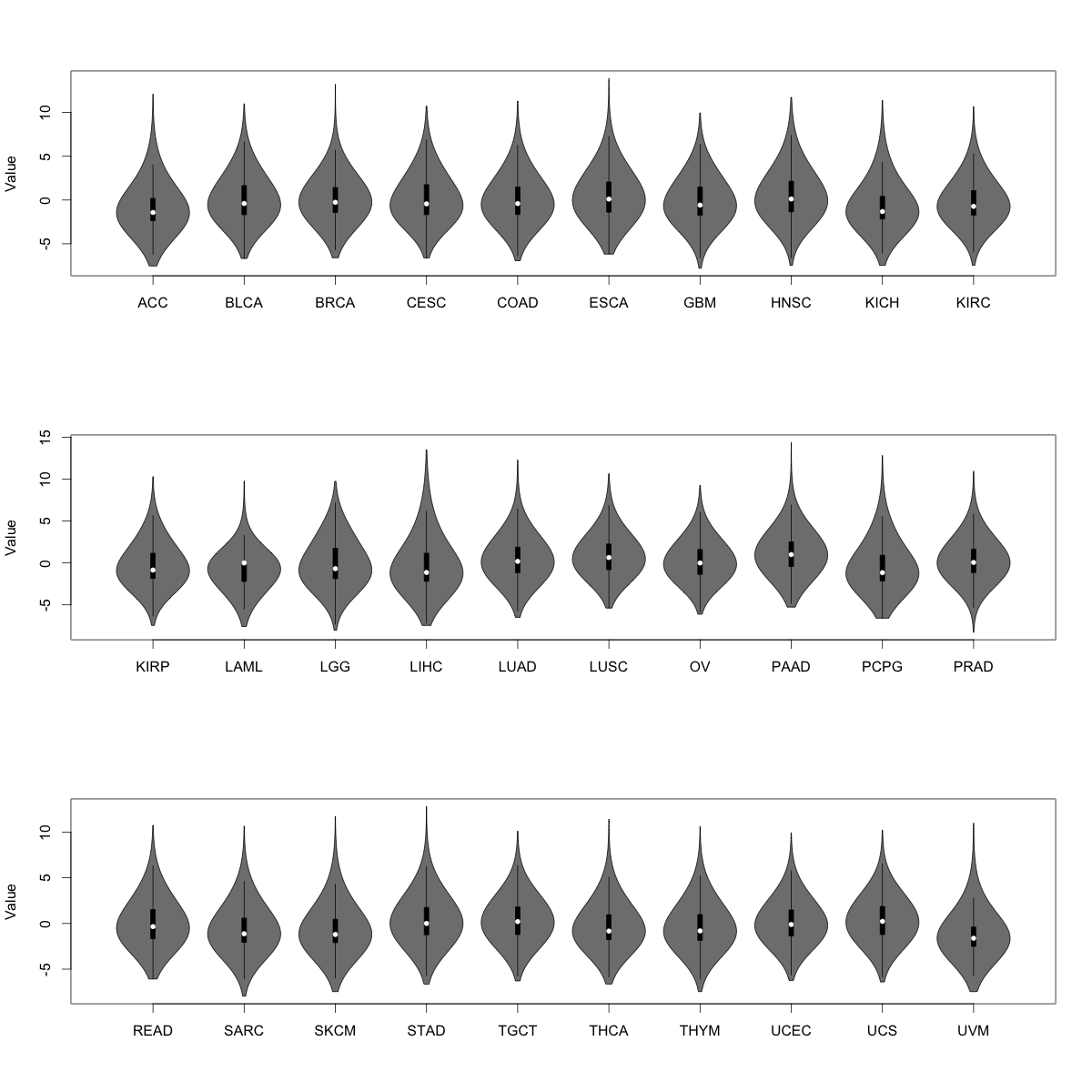}
        \caption{Violin plot of normalized expression values for each cancer type.}\label{fig:rna_violin}
\end{figure}

\begin{figure}[H] 
        \centering
        \includegraphics[width=\textwidth]{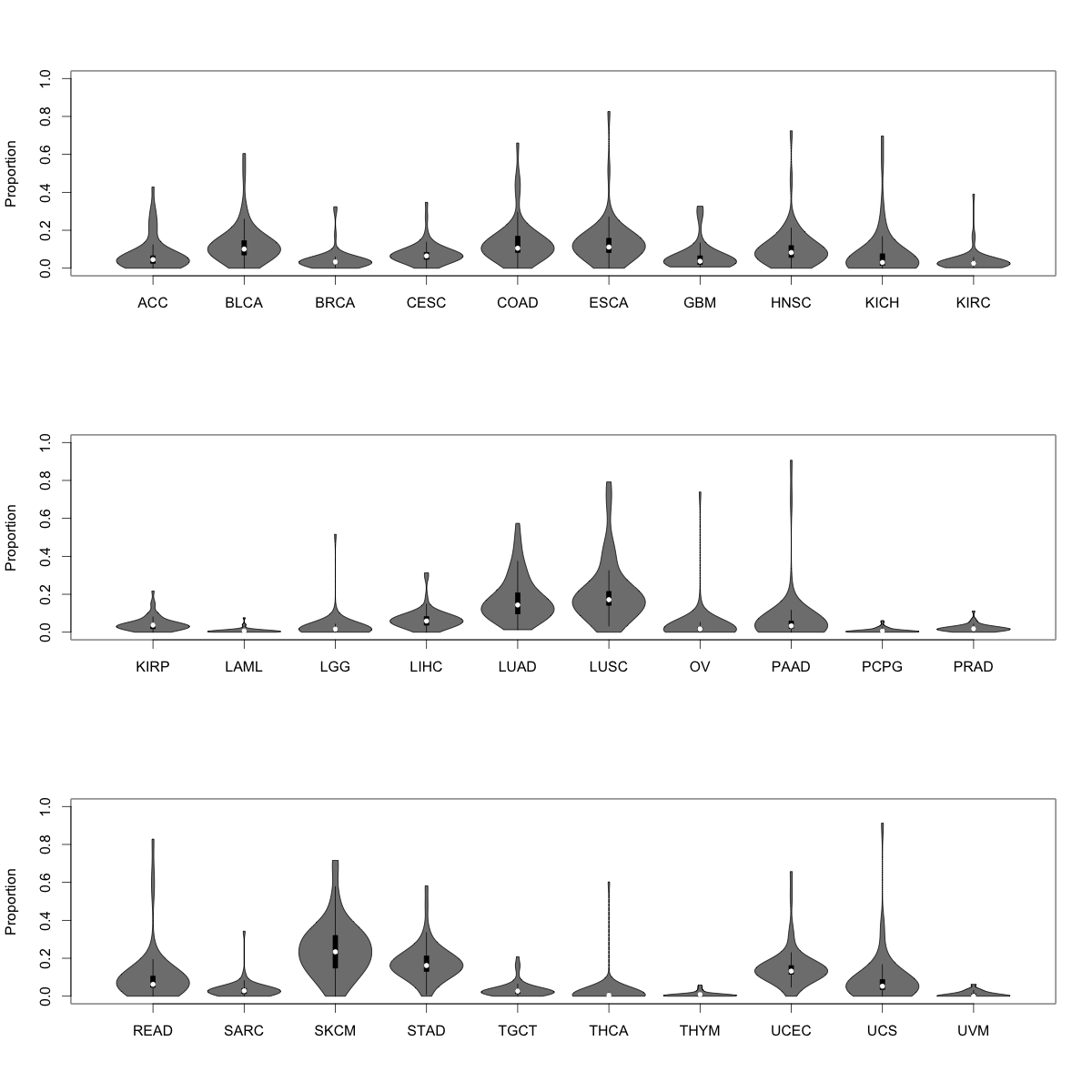}
        \caption{Violin plot of proportion of samples with a somatic mutation across the 50 genes, for each cancer type.} 
        \label{fig:mut_violin}
\end{figure}

\subsection{Selection of model parameters}
\label{app:para}

As mentioned in Section~\ref{sec:decomposition} of the main article, we first apply the optimization with dynamic modules for BIDIFAC+ \citep{Lock2022BADIFACplus}, to uncover 50 low-rank modules in $\bX_{\cdot}$. BIDIFAC+, noted in Section~\ref{sec:obejective}, is the unsupervised version of our proposed model maRRR. The statistical model of BIDIFAC+ is \begin{align*}
    \bX_{\cdot} &= \sum^{L}_{l=1}  \bS_{\cdot}^{(l)} + \bE_{\cdot} 
    \end{align*}
    where $\bS_{\cdot}^{(l)} = [\bS_{1}^{(l)},\bS_{2}^{(l)},...,\bS_{J}^{(l)}]$, $\bE_{\cdot} = [\bE_{1},\bE_{2},...,\bE_{J}]$
and 
\begin{align*}
    \bS_{j}^{(l)} = 
  \begin{dcases} 
  \boldsymbol{0}_{p\times n_j} & \text{if }  \bC_S[j,l] = 0 \\ 
  \bU_S^{(l)}\bV_{Sj}^{(l)T} & \text{if } \bC_S[j,l] = 1.
  \end{dcases}
\end{align*}The loss objective is \begin{align*}
    \min_{\{\bS_{\cdot}^{(l)}\}^L_{l=1}} \{ \frac{1}{2} ||\bX_{\cdot} - \sum^{L}_{l=1}  \bS_{\cdot}^{(l)}||_F^2 +\sum^{L}_{l=1} \lambda_S^{(l)} ||\bS_{\cdot}^{(l)}||_{*}  \}.
\end{align*}

%The model only captures all signals as $\sum^{L}_{l=1}  \bS_{\cdot}^{(l)}$ but does not distinguish whether they are covariate-related effects or not. The forward selection process of BIDIFAC+ starts with $\bC_S[:,l] = \bold{0}, \forall l = 1,...,L$, then iteratively adds cohorts $j$ ($\bC_S[j,l] = 1$) to minimize the objective function and updates the estimations by gradient descent similar to our algorithm 2. The details of module information is shown in the ensuing heatmap in Figure~\ref{fig:heatmap_CS}. The process is repeated until convergence of loss and then we take the current $\bC_S$ as final modules. We grid search the number of modules. We choose $L = 50$ since more modules than 50 does not capture additional significant amount of variances to explain. There are several modules selected by different $L$ and those modules are verified to contribute a lot of variance in the final model. Then, we choose the corresponding module indicator matrix $\bC_S$ as $L = 50$ and therefore, set $\bC_Y = \bC_S$ to partition the amount of variance related to covariate effects. 

The model aggregates all signals as $\sum^{L}_{l=1} \bS{\cdot}^{(l)}$, yet it doesn't differentiate whether these signals are related to covariates or not. The forward selection process to determine $\bC_S$ of BIDIFAC+ initiates with $\bC_S[:,l] = \bold{0}$ for all $l = 1,...,L$, and progressively includes cohorts $j$ ($\bC_S[j,l] = 1$) to minimize the objective function; see Section 6.3 in \citep{Lock2022BADIFACplus} for complete details. This update occurs iteratively through gradient descent, similar to our Algorithm 2 in Section \ref{sec:opt} but with dynamic module memberships.  This iterative process continues until convergence of loss, at which point the current $\bC_S$ is considered the final set of modules. After thorough analysis, we opt for $L = 50$ since including more modules does not significantly enhance our ability to explain variance. A few modules, selected at various $L$ values, contribute substantially to the final model's variance. Consequently, we select the corresponding module indicator matrix $\bC_S$ for $L = 50$, thus setting $\bC_Y = \bC_S$ to effectively partition the variance linked to covariate effects.
A comprehensive visualization of final estimated module information is presented in the subsequent heatmap in Figure \ref{fig:heatmap_CS}.

\begin{figure}[H] 
        \centering
        \includegraphics[height=10.5cm,width=16cm]{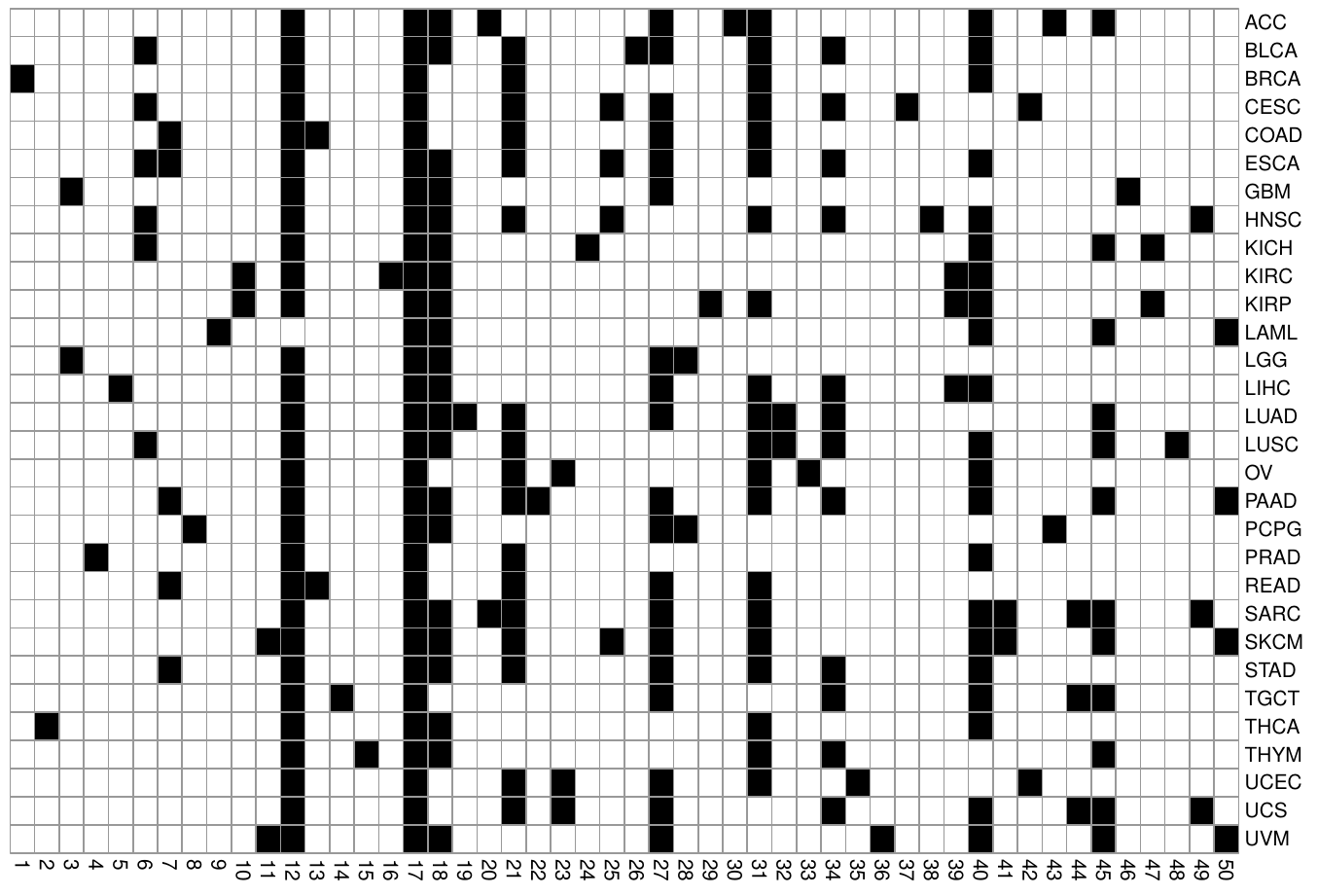}
        \caption{Heatmap for 50 modules detected by BIDIFAC+ model (the chosen $\bC_S$) in the TCGA real data application. The row names represent the 30 cohorts while the column names represent the 50 modules. Black grids mean existence of cohorts in the current module. } \label{fig:heatmap_CS}
\end{figure}

%\begin{figure}[H] \label{fig:heatmap_CS_t}
%        \centering
%        \includegraphics[height=17cm,width=11cm]{heatmap_CS.pdf}
%        \caption{eatmap for 50 modules detected by BIDIFAC+ model (the transpose of chosen $\bC_S$) in the TCGA real data application. The column names represent the 30 cohorts while the row names represent the 50 modules. Black grids mean existence of cohorts in the current module. }
       
%\end{figure}

%As described in Supporting Infomation Section D, we scale the standardized $\bX_{\cdot}$ by the median absolute deviation estimator (\citep{gavish2017optimal}) to ensure the residual variance is approximately 1, i.e. $Var(\bE_{\cdot}) = 1$. Based on random matrix theorems with variance 1 in Section 4, we assign $\lambda_B^{(i)} = \sqrt{1000} + \sqrt{50}, i = 1,...,50$ since all $\bB_k$ has the same matrix size and $\lambda_S^{(1)} = \sqrt{1000} + \sqrt{976} $, $\lambda_S^{(2)} = \sqrt{1000} + \sqrt{400} $,...,$\lambda_S^{(50)} = \sqrt{1000} + \sqrt{742} $ as the second term is the square root of the number of samples in the current module. For Algorithm 1, the general upper bound for the estimated rank of each $\bB_k$ and $\bS_\cdot^{(i)}$ is set to be 20, i.e. $r_{B,upper} = 20$ and $r_{S,upper} = 20$. We have tried different values and 20 is the least value for Algorithm 1 to match the performance of Algorithm 2. Smaller values of maximum rank will ignore some main signals while larger values will not bring significant more information to boost the prediction but just increase computational complexity.

As outlined in Appendix~\ref{app:scaling}, we apply scaling to the standardized $\bX_{\cdot}$ using the median absolute deviation estimator (\citep{gavish2017optimal}) to ensure that the residual variance is approximately 1, i.e., $Var(\bE_{\cdot}) = 1$. Drawing from random matrix theorems with a variance of 1 in Section~\ref{sec:theory}, we assign $\lambda_B^{(i)} = \sqrt{1000} + \sqrt{50}, i = 1,...,50$ since all $\bB_i$ share the same matrix size. Additionally, we set $\lambda_S^{(1)} = \sqrt{1000} + \sqrt{976}$, $\lambda_S^{(2)} = \sqrt{1000} + \sqrt{400}$, ..., and $\lambda_S^{(50)} = \sqrt{1000} + \sqrt{742}$, with the second term representing the square root of the number of samples in the respective module.

In Algorithm 1, we establish a general upper bound for the estimated rank of each $\bB_i$ and $\bS_\cdot^{(i)}$ at 20, denoted as $r_{B,upper} = 20$ and $r_{S,upper} = 20$ respectively. Through experimentation, we determined that 20 is the minimum value for Algorithm 1 to match the performance of Algorithm 2. This rank upper bound of 20 is deemed reasonable for low-rank approximations. While the true ranks of most estimates hover around 10, opting for smaller maximum rank values dismisses significant signals. Conversely, larger values don't provide substantial additional information for improving predictions, but rather increase computational complexity. 

\subsection{Additional pan-cancer analysis}
\label{app:pan_cancer}

With the identification of 50 modules and a maximum rank set at 20, maRRR attains an impressive Relative Squared Error (RSE) of 0.184, signifying the substantial capture of variation. Following Section~\ref{sec:decomposition}, the subsequent scatterplot Figure~\ref{fig:scatter} depicts the extent to which mutations account for variance within each module. While the impact of mutations is not overwhelmingly dominant, it is nevertheless significant, aligning with our initial expectations. This outcome is consistent with the extensive genetic information present in gene expressions that is unrelated to mutations. When considered alongside the heatmap of $\bC_S$ (Figure~\ref{fig:heatmap_CS} in Appendix~\ref{app:para}), it becomes evident that modules such as numbers 2, 6, 7, 13, 22, 23, 41, and 50 exhibit substantial influence from mutations. Remarkably, nearly all of the 30 cancer types make varying appearances within those modules, with the exceptions of ACC, KIRC, and PRAD.

\begin{figure}[H] 
        \centering
        \includegraphics[height=10.5cm,width=16cm]{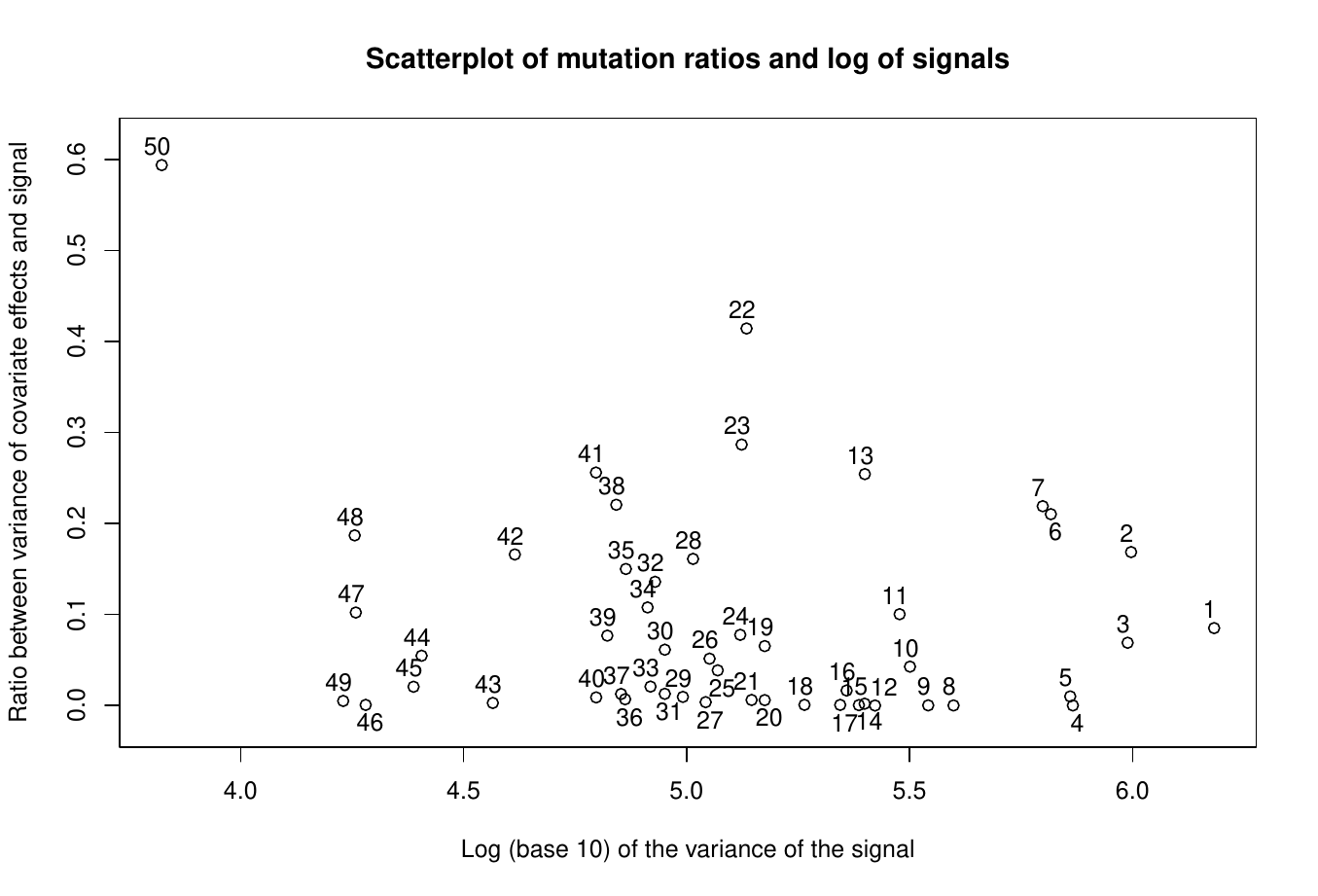}
        \caption{Scatterplot of The x-axis is the log (base 10) of the variance of the signal; the y-axis is the ratio between variance of covariate effects $\bB_i\bY_{\cdot}^{(i)}$ and signal. The number near each dot is the module identification as in $\bC_S$. }
       \label{fig:scatter}
\end{figure}

Besides the individual and partially shared shared structures, we are able to detect global shared effects as well. Module 12 has little covariate effects, which implies that all 50 candidate mutations do not have significant global effects on 29 cancer types. However, based on Figure~\ref{fig:s12}, we are able to observe several clusters that share across all the cancer types (horizontally). 

\begin{figure}[H]
        \centering
        \includegraphics[height=10cm,width=16cm]{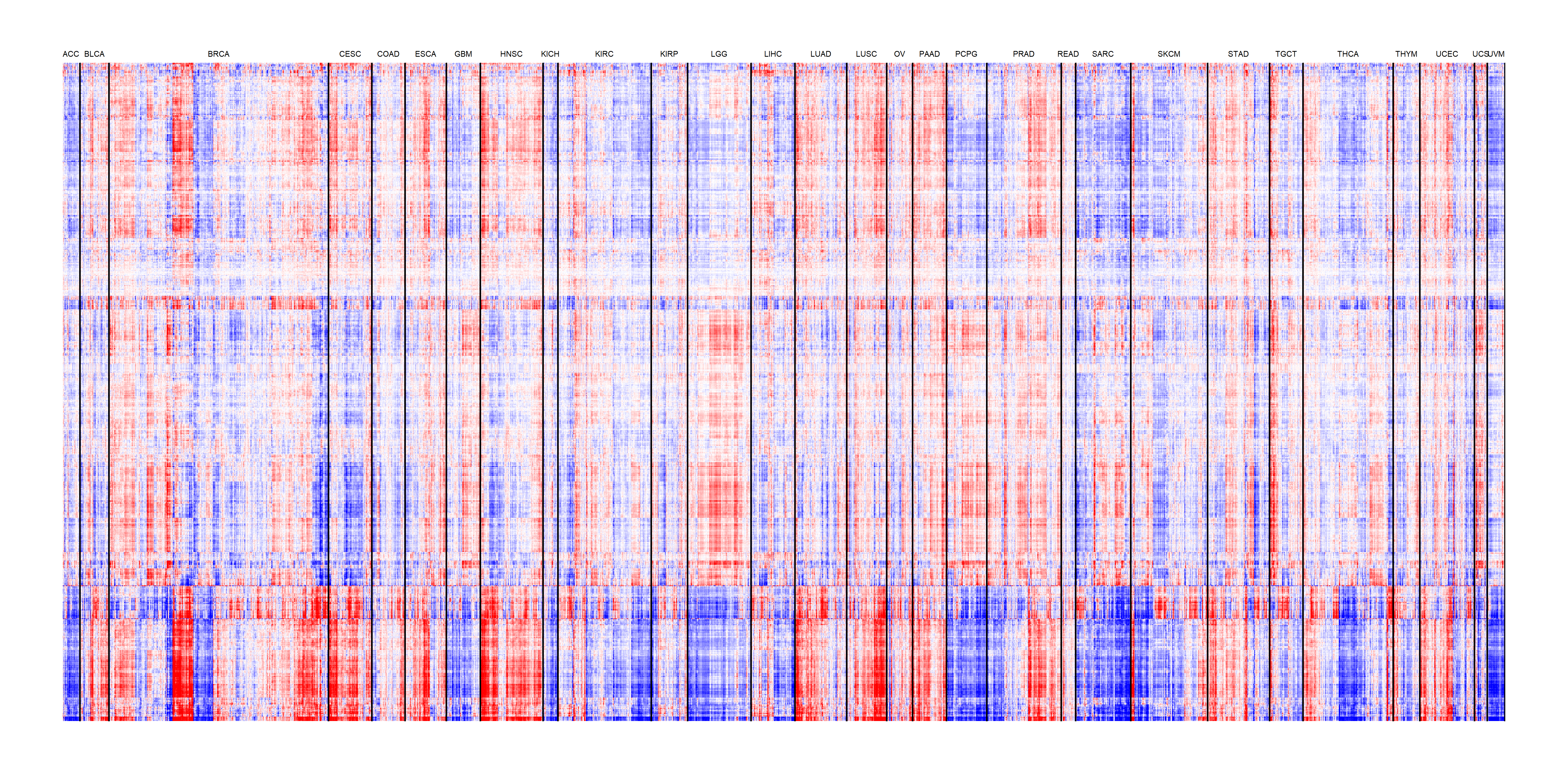}
        \caption{Heatmap for Module 12 auxiliary structure. Columns represent individual samples, while rows represent distinct gene expressions. This column-row representation is consistent across all other heatmaps in the \href{https://www.dropbox.com/s/891gvukhymr4a4b/30grps_50mods_est_0509.csv?dl=0}{online application results spreadsheet}. Extreme values outside of 3 standard deviations are set to be threshold values. The graphics is based on the relative scale. Red colors represent high gene expressions and blue colors represents low gene expressions. }
     \label{fig:s12}   
\end{figure}
\end{appendices}

\bibliographystyle{biom}
%\bibliography{references}

\label{lastpage}
\end{document}